\documentclass[a4paper,UKenglish,cleveref,autoref, thm-restate]{lipics-v2021}

\title{Online Firefighting on Cactus Graphs} 

\author{Max {Hugen}}{Department of Information and Computing Sciences, Utrecht University, The Netherlands}{}{}{}
\author{Bob {Krekelberg}}{Department of Information and Computing Sciences, Utrecht University, The Netherlands}{b.h.a.f.krekelberg@uu.nl}{https://orcid.org/0009-0000-5517-6095}{}

\author{{Alison Hsiang-Hsuan} {Liu}}{Department of Information and Computing Sciences, Utrecht University, The Netherlands}{h.h.liu@uu.nl}{https://orcid.org/0000-0002-0194-9360}{}

\authorrunning{M. Hugen, B. Krekelberg and A.\,H.\,H. Liu} 

\Copyright{Max Hugen, Bob Krekelberg and Alison Hsiang-Hsuan Liu} 



\category{} 

\relatedversion{} 

\nolinenumbers 

\EventEditors{John Q. Open and Joan R. Access}
\EventNoEds{2}
\EventLongTitle{42nd Conference on Very Important Topics (CVIT 2016)}
\EventShortTitle{CVIT 2016}
\EventAcronym{CVIT}
\EventYear{2016}
\EventDate{December 24--27, 2016}
\EventLocation{Little Whinging, United Kingdom}
\EventLogo{}
\SeriesVolume{42}
\ArticleNo{23}

\usepackage{amsfonts}
\usepackage{amsmath}
\usepackage{graphicx}
\usepackage[dvipsnames]{xcolor}
\usepackage{amsthm}
\usepackage[vlined, ruled, boxed, noend, linesnumbered]{algorithm2e}
\usepackage{xspace}
\usepackage{wrapfig}

\begin{document}
\newcommand{\todo}[1]{{\color{Red}TODO: #1}}
\newcommand{\hide}[1]{}
\newcommand{\maxt}[1]{{\color{Blue}#1 in Max's thesis}\xspace}
\newcommand{\sketch}[1]{\begin{proof}[Proof sketch]#1\end{proof}}
\newcommand{\runtitle}[1]{\textbf{#1}}

\newcommand{\alg}{\text{ALG}\xspace}
\newcommand{\algAlmost}{\text{ALG}_{\texttt{A}}}
\newcommand{\algCactus}{\text{ALG}_{\texttt{C}}}
\newcommand{\opt}{\text{OPT}\xspace}
\newcommand{\algset}{\Psi^\text{A}}
\newcommand{\optset}{\Psi^*}
\newcommand{\instance}{\mathcal{I}}

\newcommand{\cool}{\texttt{Cool-down}}
\newcommand{\uCutCycle}{\hat{u}}
\newcommand{\CCutCycle}{\hat{C}}
\newcommand{\etol}{\texttt{tol}_\texttt{e}}
\newcommand{\EnewT}[2]{T_e(#1\backslash(#2))}
\newcommand{\umax}{u_*}
\newcommand{\Cmax}{\rootC_*}
\newcommand{\cactusBreakCycle}{\textsc{ImprovedBreak}\xspace}

\newcommand{\delayedSet}{\texttt{D}}
\newcommand{\optDelayedSet}{\texttt{D}_{\optset}}

\newcommand{\valg}{\psi^\texttt{A}}
\newcommand{\vopt}{\psi^*}
\newcommand{\Galg}{G^\texttt{A}}

\newcommand{\CoveredSet}{\kappa}
\newcommand{\exclusiveCoveredSet}{\kappa^\prime}
\newcommand{\weight}{w}
\newcommand{\exclusiveWeight}{w^\prime}
\newcommand{\dist}{\texttt{dist}}
\newcommand{\amountBelow}{count}
\newcommand{\tol}{\texttt{tol}}
\newcommand{\newT}[2]{T(#1\backslash#2)}

\newcommand{\newG}{\widetilde{G}}
\newcommand{\newR}{\tilde{r}}
\newcommand{\rootC}{\widetilde{C}}

\newcommand{\charging}{\chi}
\newcommand{\partition}{\mathcal{P}^*}
\newcommand{\parta}{P^{(a)}}
\newcommand{\partb}{P^{(b)}}
\newcommand{\partc}{P^{(c)}}

\newcommand{\dmax}{d_{\max}}

\newcommand{\tHeaviest}{v^{(t)}_1}
\newcommand{\sHeaviest}{v^{(s)}_1}

\newcommand{\algEven}{\text{ALG}_{\texttt{E}}}
\newcommand{\newWeight}{\bar{\weight}}
\maketitle

\begin{abstract}
    The \emph{firefighting game} is a fundamental problem in theoretical computer science, modeling the containment of a spreading process under limited defensive resources.
    In each round, an algorithm may protect a set of vertices by placing firefighters on them, preventing the fire from spreading to those vertices. 
    Vertices that are never reached by the fire are said to be \emph{saved}.
    The goal is to maximize the number of saved vertices.
    While the offline version has been extensively studied, much less is known about the \emph{competitive complexity} of the online variant, where the underlying graph is known in advance but the number of available firefighters in each round is revealed online.

    We study how \emph{graph structure} governs the power of the adversary in the online firefighting game.
    On trees, the problem is known to be $2$-competitive (Coupechoux et al., 2019), a result that relies on a strong structural alignment between the online algorithm and the optimal solution throughout the process.
    We show that this alignment breaks down as soon as cycles are present.
    In particular, the algorithm and the optimal solution may break a cycle differently, or even break different cycles, and therefore operate on fundamentally different residual graphs.
    As a result, the adversary gains additional leverage by steering the process along residual graph structures that no longer admit a direct comparison between the algorithm and the optimal solution.

    Our main result shows that the presence of a single cycle already increases the competitive complexity of the firefighting game dramatically.
    We first show that even on a tadpole graph (a cycle with a tail), no deterministic online algorithm can achieve a competitive ratio better than $\Omega(\sqrt{n})$, where $n$ is the number of vertices.
    We complement this lower bound with matching upper bounds by designing an $O(\sqrt{n})$-competitive online algorithm for $1$-almost trees, that is, graphs obtained from a tree by adding at most one edge.
    Furthermore, we extend our framework to cactus graphs and prove that, despite the presence of multiple cycles, the competitive complexity remains $\Theta(\sqrt{n})$ as long as the cycles do not share edges.
    Together, these results yield a tight characterization of the adversarial power induced by cycles with non-overlapping edges.
    Finally, considering that cactus graphs have treewidth of $2$, we study a variant in which firefighters are released in pairs, that is, an even number of firefighters becomes available in each round.
    Surprisingly, the competitive complexity is significantly reduced to $3$-competitive for this setting.

    The main technical challenge lies in the analysis, as the algorithm and the optimal solution may break different cycles and therefore operate on different residual graphs.
    As a result, unlike the tree case, it is no longer possible to rely on a direct correspondence between saved vertices in the two solutions.
    To address this, our algorithm makes sure that either it saves a large enough portion of the vertices that the optimal solution can save in the same round, or it delays the timing of the fire reaching a large portion of vertices by breaking a carefully selected cycle at a meticulously selected position.
    For the analysis, we conceptually match the vertices where the algorithm places firefighters with the vertices protected by the optimal solution at the corresponding rounds.
    We then partition the vertices protected by the optimal solution so that their total contribution can be charged to the vertices protected by the algorithm, with each algorithmic vertex being charged only a constant number of times.
\end{abstract}

\newpage

\section{Introduction}
The \emph{online firefighting game} proceeds in rounds as follows. 
Initially, the algorithm is given a graph with a designated fire source.
At the beginning of each round, the algorithm is informed of the number of firefighters available in that round and assigns them to protect unburned vertices.
Subsequently, every unprotected neighbor of a burning vertex catches fire, concluding the round.
The online algorithm aims to maximize the number of saved vertices, without knowing in advance how many firefighters will be available in future rounds.

We evaluate the performance of an online algorithm for the firefighting game using \emph{competitive analysis}.
An algorithm $\alg$ is $c$-competitive if $\sup_\instance\frac{\opt(\instance)}{\alg(\instance)} \leq c$, where $\alg(\instance)$ is the number of vertices saved by $\alg$ on instance $\instance$, and $\opt(\instance)$ is the maximum number of vertices that can be saved with full knowledge of the firefighter sequence.
From the complexity perspective, the competitive ratio of an online problem is defined as the infimum competitive ratio achievable by any online algorithm for the problem.

\runtitle{Related work.}
Introduced by Hartnell in 1995~\cite{hartnell1995firefighter}, the firefighter problem can be seen more generally as a model for a threat spreading through a network.
For example, the problems of minimizing the spread of misinformation or threats in social or computer networks, immunizing a population against a virus, and reducing the spread of an invasive species in an ecosystem are also captured by this model.

The firefighting game is known to be a computationally hard problem.
Even for restricted graph classes such as bipartite graphs~\cite{macgillivray2003firefighter} and even for trees of maximum degree~$3$~\cite{DBLP:journals/dm/FinbowKMR07}, the problem is shown to be NP-hard.

On the positive side, several approximation algorithms are known for the firefighting game.
For example, the greedy algorithm is shown to be a $2$-approximation on trees~\cite{hartnell2000firefighting}, and also, there also exists a polynomial time approximation scheme for the firefighting game on trees~\cite{DBLP:journals/talg/AdjiashviliBZ19}.

The online variant of the firefighting game was introduced by Coupechoux et al.~\cite{DBLP:journals/tcs/CoupechouxDEJ19}.
They showed that the greedy algorithm on trees in the online version of the problem is $2$-competitive and optimal.
Also, in the fractional version of the problem, in which a fraction of a firefighter can be assigned to a vertex to partially protect it, Coupechoux et al.~\cite{DBLP:journals/tcs/CoupechouxDEJ19} show that the greedy algorithm is $2$-competitive on trees.
Furthermore, in the special case where the number of firefighters is bounded by $2$, Coupechoux et al.~\cite{DBLP:journals/tcs/CoupechouxDEJ19} propose a $\phi$-competitive algorithm on trees, where $\phi$ is the golden ratio.
Later, Demange et al.~\cite{DBLP:conf/ictcs/DemangeEG19} investigated sufficient conditions for online containment of fires on infinite grid graphs.
They show that for an algorithm to contain the fire, there needs to exist a turn $t$ such that the sum of firefighters released until (including) turn $t$ is at least $16t$.

\medskip
\begin{wrapfigure}{R}{0.3\linewidth}
    \centering
    \includegraphics[width=0.75\linewidth]{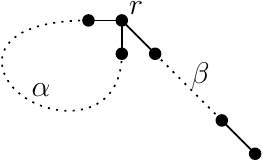}
    \caption{An $(\alpha +1, \beta + 1)$-tadpole graph.}
    \label{fig:tadpole}
\end{wrapfigure}

\runtitle{Our contribution and techniques.}
In this work, we study the competitive complexity of the online firefighting game.
It is known that the game is $2$-competitive on trees, indicating that the problem is relatively simple in this setting.
This raises the question of how the competitive complexity evolves on broader classes of graphs.
Surprisingly, we find that with the addition of a single extra edge to a tree, the problem is no longer constant-competitive.

\begin{theorem}\label{thm:TadpoleLB}
    No deterministic online algorithm can achieve $o(\sqrt{n})$-competitiveness for the online firefighting game on tadpole graphs.
\end{theorem}

\begin{proof}
    Consider an arbitrary deterministic online algorithm \alg.
    Let the input graph $G$ be an $(\alpha+1, \beta+1)$-tadpole graph, consisting of a cycle of $\alpha+1$ vertices and a path of $\beta+1$ vertices joined at a common vertex $r$, where $r$ is the unique degree-$3$ vertex of $G$. 

    \runtitle{Case 1.}
    In the first round, the adversary releases one firefighter. 
    If \alg protects any vertex on the cycle, the adversary stops releasing more firefighters. 
    In this case, \opt protects the neighbor of $r$ on the path and saves $\beta$ vertices of the path (where $r$ is burned), whereas $\alg$ saves only the single protected vertex.

    \runtitle{Case 2.}
    If \alg protects a vertex on the path in the first round, then the adversary releases one more firefighter in the next round and then stops. 
    In this case, \alg can save at most $1 + \beta$ vertices (the initial protection and the path).
    Meanwhile, \opt can protect two vertices on the cycle: one neighbor of $r$ in the first round, and the vertex adjacent to the other neighbor of $r$ on the cycle in the second round.
    Thus, $\opt$ saves at least $\alpha-1$ vertices on the cycle.
    
    Therefore, the competitive ratio is at least $\frac{\opt((G,r,(f)_{i\geq 1}))}{\alg((G,r,(f)_{i\geq 1}))} \geq \min \{\frac{\beta}{1}, \frac{\alpha-1}{1+\beta}\}$.
    To maximize this lower bound, we balance the two terms by setting $\beta = \frac{\alpha - 1}{\beta + 1}$, which yields $\alpha = \beta(\beta+1)+1$.
    Since the number of vertices in $G$, $n= \alpha+\beta+1 = \Theta(\beta^2)$, it follows that the competitive ratio of $\alg$ is at least $\Omega(\sqrt{n})$.
\end{proof}

We then design an $O(\sqrt{n})$-competitive algorithm for $1$-almost trees, that is, graphs formed by adding at most one extra edge to a tree.
Equivalently, a $1$-almost tree contains at most one cycle, thereby generalizing tadpole graphs.
More specifically:

\begin{theorem}\label{thm:Almost}
    There is a $(5\sqrt{n}+2)$-competitive algorithm $\algAlmost$ for the online firefighting game on $1$-almost-trees, which is asymptotically optimal.
\end{theorem}

By further generalizing our algorithm to cactus graphs with multiple cycles, we find that the number of cycles does not increase the difficulty of the problem from the perspective of competitive complexity.
Specifically:
\begin{theorem}\label{Thm:Cactus}
    There is a $(20\sqrt{n}+1)$-competitive algorithm $\algCactus$ for the online firefighting game on cactus graphs, which is asymptotically optimal.
\end{theorem}

Since cactus graphs have treewidth at most two, any cactus graph admits a tree decomposition in which each bag has size at most three.
A natural question then arises: what happens if firefighters are released in a multiple of the treewidth?
Surprisingly, we find that the special instance on cactus graphs, where an even number of firefighters is available in each round, significantly reduces the competitive complexity of the firefighting game:
\begin{theorem}\label{Thm:Even}
    There is a $3$-competitive algorithm $\algEven$ for the online firefighting game on cactus graphs with even firefighters in each round.
\end{theorem}

\runtitle{Technical challenge.}
A natural approach is to adapt the greedy algorithm for trees in~\cite{DBLP:journals/tcs/CoupechouxDEJ19}.
In trees, although the greedy algorithm and the optimal solution may protect different vertices in each round, the tree structure guarantees that there is a unique path from the fire source to each vertex.
Thus, any vertex available to the optimal solution is also available to the algorithm, unless it has already been saved by the algorithm.
The analysis then charges the vertices saved exclusively by the optimal solution to those saved by the algorithm, according to the greedy nature of the algorithm.
Additionally, the number of vertices protected by both the optimal solution and the algorithm is at most the number of vertices saved by the algorithm, which implies that the optimal solution saves at most twice as many vertices as the algorithm.

However, the availability guarantee no longer holds when the graph contains a cycle.
In this case, the shortest distance between a vertex and the fire source (i.e., the earliest time the vertex can be burned) can be increased by ``breaking'' a cycle that overlaps the shortest path from the vertex to the fire source.
Consequently, if the online algorithm breaks the cycle differently from the optimal solution, the two strategies result in different residual graphs, which invalidates the charging framework.

In cactus graphs, even multiple cycles can be present.
This introduces further misalignments between the algorithm’s strategy and that of the optimal solution: not only in how to break a cycle, but also in which cycle to break.

The misalignment between the algorithm’s and the optimal solution’s residual graphs also makes the analysis challenging.
We design a charging framework, which first partitions the vertices saved by the optimal solution into independent parts, and then charges each part to a subset of vertices saved by the algorithm.
The number of parts may grow depending on the configuration of the optimal solution.
Thus, the challenge is to partition the vertices saved by the optimal solution appropriately (based on different graph classes and algorithms) and to select the subsets of vertices saved by the algorithm such that no vertex saved by the algorithm is charged more than a constant number of times.

\runtitle{Paper organization.}
In Section~\ref{Sec:Framework}, we introduce the key components of our algorithmic design, provide a sketch of the algorithms, and present the analysis framework together with the properties that ensure its feasibility.
In Section~\ref{Sec:almostTree}, we demonstrate in detail how to apply our design and analysis framework to $1$-almost trees and prove Theorem~\ref{thm:Almost}.
Section~\ref{Sec:cactus} further extends the framework to cactus graphs, establishing Theorem~\ref{Thm:Cactus}.
Finally, in Section~\ref{Sec:even}, we adapt the framework to a special case on cactus graphs in which firefighters arrive in pairs, and show that the competitive complexity in this setting is constant.

\section{Preliminaries, algorithmic blueprint, and analysis framework}
\label{Sec:Framework}
Let $G=(V, E)$ be an undirected graph. 
A \emph{(simple) cycle} is a sequence of distinct vertices $(u_1, u_2, \cdots, u_k) \subseteq V$ with $k \geq 3$ such that $(u_i, u_{i+1}) \in E$ for all $1 \leq i \leq k-1$, and $(u_1, u_k) \in E$.\footnote{For accommodating other notations used in the paper, we use $(u, v)$ as edges. Since the graph discussed in this work is undirected, $(u, v) = (v, u)$.}
Given any graph $G$, we denote $V(G)$ as the set of vertices in $G$ if it is not specified.
The graph $G$ is a \emph{$1$-almost-tree} if it is isomorphic to a connected graph with $|V(G)|$ vertices and $|V(G)|$ edges. 
Equivalently, a $1$-almost-tree can be seen as a tree with one extra edge and has exactly one cycle.
The graph $G$ is a \emph{cactus} graph if every edge $e\in E$ is contained in at most one cycle. 
In other words, in cactus graphs, two cycles can share at most one vertex.
In this paper, we say that a vertex is a \emph{cycle vertex} if it is in at least one cycle, and a \emph{non-cycle} vertex otherwise. 
Moreover, a cycle is a \emph{root cycle} if it contains the fire source $r$.

In the online firefighting game, the instance is $\mathcal{I} = (G, r, (f_i)_{i\geq 1})$, where $G = (V, E)$ is a graph with the \emph{fire source} or \emph{root} $r$, and $f_i \geq 0$ is the number of available firefighters in round $i$. 
We denote $n$ as the number of vertices in $G$ (that is, $n = |V(G)|$ using our notation).
Let the \emph{(open) neighborhood} of a vertex $v\in V$ be the set of all vertices adjacent to $v$ and denoted by $N(v)$. That is, $N(v) = \{u\mid (u,v) \in E\}$.
In each round $i$ of the firefighting game, the online algorithm learns the value of $f_i$ and \emph{protects} unburned vertices by assigning the firefighters to them.
Then, the fire spreads to unprotected vertices in $N(v)$ for all burning vertices $v$, and it ends this round.
In this work, we consider the cases where $G$ is a $1$-almost-tree (Section~\ref{Sec:almostTree}) or a cactus graph (Section~\ref{Sec:cactus}).

\medskip

\runtitle{Status of a vertex.}
At any time in the firefighting game, a vertex can be \emph{protected}, \emph{burned}, \emph{saved}, or \emph{available}.
A vertex is \emph{protected} if there is a firefighter assigned to it.
A vertex $v$ is \emph{burned} if there is no protected vertex on any path from $r$ to $v$.
Conversely, a vertex is \emph{saved} if there is at least one protected vertex on every path from $r$ to $v$.
In the middle of the game, a vertex may not be burned yet, but there are some paths between it and the fire source $r$ that still have no protected vertices yet. In this case, we say that this vertex is \emph{available}.
The algorithm's goal is equivalent to saving as many vertices as possible.

\medskip

\begin{figure}
     \centering
     \begin{subfigure}[b]{0.3\textwidth}
         \centering
         \includegraphics[width=\textwidth]{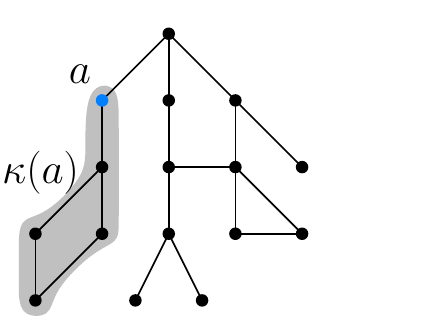}
         \caption{The covered set of a protected non-cycle vertex.}
         \label{fig:covered_set_node}
     \end{subfigure}
     \hfill
     \begin{subfigure}[b]{0.3\textwidth}
         \centering
         \includegraphics[width=\textwidth]{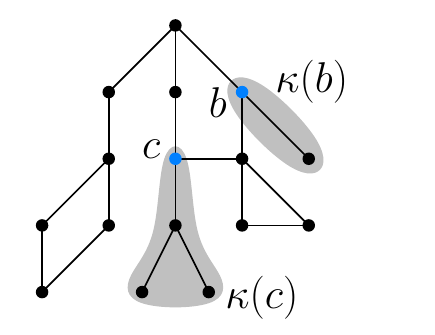}
         \caption{The covered sets of protected cycle vertices.}
         \label{fig:covered_set_cycle_node}
     \end{subfigure}
     \hfill
     \begin{subfigure}[b]{0.3\textwidth}
         \centering
         \includegraphics[width=\textwidth]{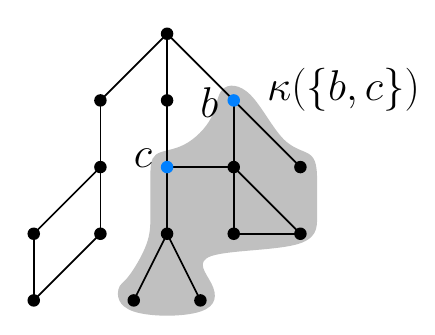}
         \caption{The covered set of a set of protected (cycle) vertices.}
         \label{fig:covered_set_set}
     \end{subfigure}
        \caption{Examples of covered sets.}
        \label{fig:covered_set}
\end{figure}

\runtitle{Covered set and profit of algorithms.}
A vertex $u$ is \emph{covered} by $v$ if every path from $u$ to $r$ contains vertex $v$.
Conversely, all the vertices covered by $v$ form a \emph{covered set} of $v$, denoted by $\CoveredSet(v)$.
See Figure~\ref{fig:covered_set_node}.
More generally, a vertex $u$ is covered by a set of vertices $S$ if every path from $u$ to $r$ contains at least one vertex in $S$, and the covered set of $S$, $\CoveredSet(S)$, is defined analogously.\footnote{Naturally, we only consider $S \subseteq V\setminus \{r\}$.}
Intuitively, the covered set of a set of vertices $S$ can be seen as the set of vertices that are saved by protecting all vertices in $S$ (before fire spreads to them). 
Equivalently, it is the set of vertices that are saved by placing firefighters on each vertex in $S$ before the fire starts spreading. 
Note that by this definition, $\CoveredSet(S)$ can be larger than $\bigcup_{v\in S} \CoveredSet(v)$, as illustrated in Figure~\ref{fig:covered_set_cycle_node} and Figure~\ref{fig:covered_set_set}.
The \emph{weight} of a set of vertices $S$, $\weight(S) = |\CoveredSet(S)|$, is the size of its covered set. 

Let $\algset$ denote the set of all vertices protected by our algorithm \alg.
By definition, $\CoveredSet(\algset)$ is the set of vertices saved by \alg, and $\weight(\algset)$ is the final \emph{profit} of \alg. That is, $\weight(\algset)$ is the total number of vertices saved by \alg.
Similarly, $\optset$ denotes the set of all vertices protected by the optimal solution \opt, and $\CoveredSet(\optset)$ and $\weight(\optset)$ are defined symmetrically. 

\subsection{Algorithmic blueprint}
In this subsection, we first introduce essential notions for the algorithm, and then we provide the blueprint of the algorithms that will be further discussed in detail in later sections.

\paragraph*{Notions for the algorithm}

\begin{figure}[t]
    \centering
    \includegraphics[width=\linewidth]{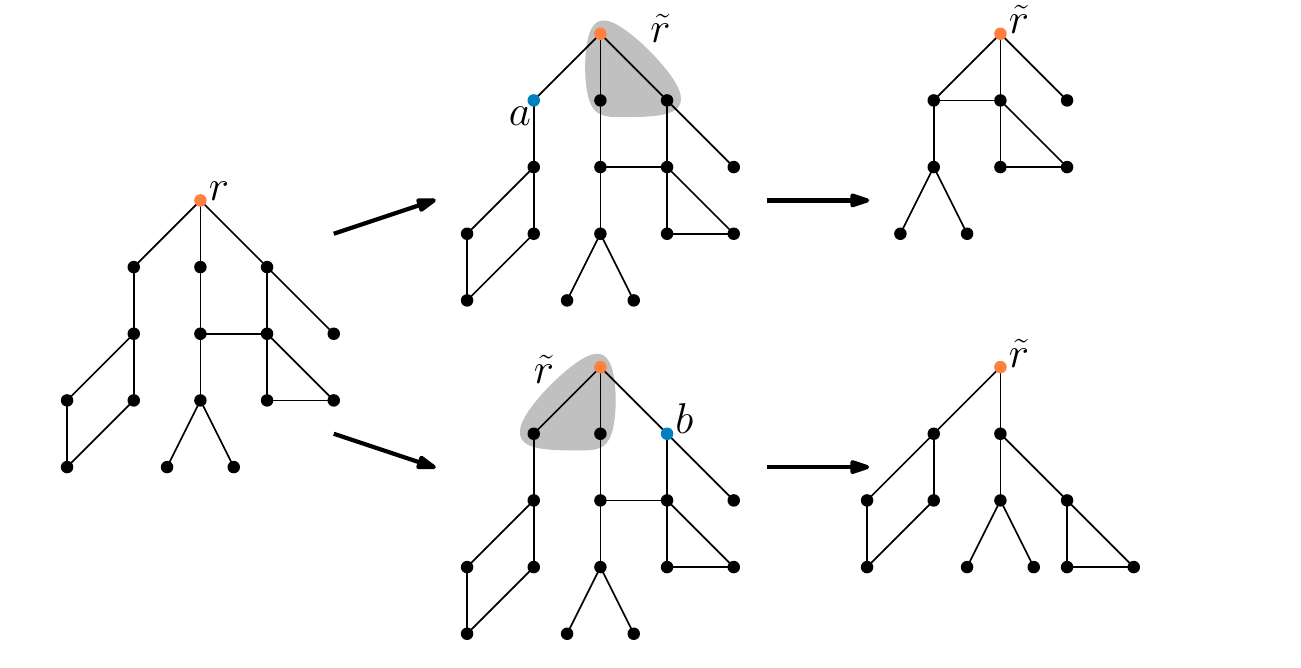}
    \caption{Graph reduction. The left-hand figure shows the graph at the beginning of the round, with the orange vertex as the fire source. The middle figure illustrates the spread of the fire to unprotected vertices by the end of the round where the blue vertices $a$ and $b$ are protected. The right-hand figure shows the graph after the graph reduction.
    \label{fig:graph_change_alg}}
\end{figure}

\runtitle{Graph reduction.}
We follow the reduction procedure on online firefighting on trees by Coupechoux et al.~\cite{DBLP:journals/tcs/CoupechouxDEJ19}. 
At the beginning of a round, the algorithm sees the graph $G$ with a fire source $r$.
At the end of the round, after the algorithm's decisions of protecting a set of available vertices $P$, and the fire spreading to all unprotected neighbors of $r$, all burning vertices are contracted with $r$ into a new fire source $r^\prime$. 
Moreover, all vertices in $\CoveredSet(P)$ are removed from the graph. 
The new graph $G^\prime$ with the new fire source $r^\prime$ is the input to the algorithm in the next round.
By this reduction procedure, the algorithm always sees exactly one burning vertex, which is the fire source, and all other vertices are available for being protected.
An example of the reduction procedure is illustrated in Figure~\ref{fig:graph_change_alg}.

Note that the resulting graph is sensitive to the protected vertex. 
In the case where the algorithm protects a different vertex compared to the optimal algorithm, the graph resulting from the reduction procedure for each algorithm might differ substantially.
For an illustration of this, compare the difference after protecting vertex $a$ and vertex $b$ in Figure~\ref{fig:graph_change_alg}.

In our work, we investigate graphs that are \emph{$1$-almost-trees} or \emph{cactus graphs}. 
By a slight abuse of terminology, we henceforth define a $1$-almost-tree as a graph containing \emph{at most} one cycle.
Under this definition, both the class of $1$-almost-trees and the class of cactus graphs are closed under vertex extraction and under contraction of a vertex together with some of its neighbors.
Consequently, it suffices to analyze the algorithm in a single round, as the same arguments extend to the full game.

\medskip

\runtitle{Subgraphs induced by covered sets.}
Protecting a cycle vertex ``breaks'' the cycle, transforming it into tree branches in the reduced graph.
See Figure~\ref{fig:graph_change_alg} for an example where the protected vertex $b$ breaks the $5$-cycle.
Given a root cycle $\rootC$ and a vertex $v$ on it, we denote $\newT{\rootC}{v}$ as the subgraph induced by $\{r\} \cup \CoveredSet(\rootC)$ with $v$ being protected and $\CoveredSet(v)$ being removed. 
Formally, $\newT{\rootC}{v} = G[\{r\} \cup \CoveredSet(\rootC) \setminus \CoveredSet(v)]$.
We extend this notation to $\EnewT{\rootC}{u,v}$, where $(u,v)$ is an edge of the root cycle $\rootC$, to denote the subgraph induced by $\{r\} \cup \CoveredSet(\rootC)$ with the edge $(u,v)$ removed.

\medskip

\runtitle{Temporal safeness and tolerance.}
Let $\dist(G, u, v)$ denote the shortest distance between vertices $u$ and $v$ in graph $G$. 
Our algorithm heavily relies on two measurements: the number vertices in $G$ that are \emph{temporally safe} with respect to $d$, $\amountBelow(G, d)$, and the \emph{tolerance} of a cycle vertex $u$ in a root cycle $C$ with respect to a natural number $m$, $\tol(u, C, m)$.
Formally, $\amountBelow(G, d)$ is the number of vertices in $G$ that have the shortest distance to $r$ of at least $d$. 
That is, $\amountBelow(G, d) = |\{v\in V \mid \dist(G, r, v) \geq d\}|$.
Note that, by definition, the function $d \mapsto \amountBelow(G,d)$ is monotone non-increasing.
More specifically, $\amountBelow(G, d_2) \leq \amountBelow(G, d_1)$ for any $0 \leq d_1 \leq d_2$ for all cactus graphs $G$.
The notion of tolerance $\tol(u, C, m)$ is defined as the largest integer $d$ such that, in the induced subgraph $\newT{C}{u}$, there are at least $m$ vertices whose distance from the root is at least $d$.
Formally, $\tol(u, C, m) = \max\{d \in \mathbb{N} \mid \amountBelow(\newT{C}{u},d) \geq m\}$.
Intuitively, the tolerance $\tol(u, C, m)$ measures how far the fire can spread in the subgraph $\newT{C}{u}$ while still leaving at least $m$ vertices temporally safe.
Given an edge between two cycle vertices $u_i$ and $u_j$, we also extend the notion to $\etol((u_i,u_j), C, m) = \max\{d \in \mathbb{N} \mid \amountBelow(\EnewT{C}{u_i, u_j},d) \geq m\}$.

\paragraph*{Blueprint of the algorithms}
In each round, our algorithms assign the available firefighters through a sequence of iterations.
That is, consider all connected components adjacent to the fire source.
In each iteration, the algorithm greedily assigns firefighters to vertices in order to protect one connected component at a time.
If there are enough firefighters available to save the heaviest component, the algorithm saves the heaviest component by protecting one or two vertices, depending on the structure of the component (in particular, on whether it contains a root cycle, as is the case for $1$-almost trees).
The iteration then terminates, and the algorithm proceeds to the next iteration within the same round if there are still unassigned firefighters.

The situation becomes non-trivial when there are not enough firefighters to protect the heaviest component. 
In $1$-almost trees or cactus graphs, this can only occur when exactly one firefighter is available, and the heaviest component has exactly two vertices adjacent to the fire source (which contains a root cycle $\rootC$).

Inspired by Theorem~\ref{thm:TadpoleLB}, our algorithm for $1$-almost trees proceeds as follows. 
If there exists a component whose weight is at least $\sqrt{\weight(\rootC)}$, the algorithm protects this heaviest component.
Otherwise, the algorithm gives up immediate gain and instead \emph{breaks the cycle} and modifies the graph structure to ensure that whenever the optimal solution can save many vertices in the future rounds, the algorithm is also able to save a sufficient number of vertices.
Concretely, the algorithm protects the cycle vertex that delays the earliest time at which a significant number (that is, at least $\sqrt{\weight(\rootC)}$) of vertices would be burned for as long as possible, in the hope of securing a future gain.

The cycle-breaking procedure is more involved in cactus graphs.
On the one hand, since there might be multiple root cycles at the same time, the algorithm must resolve ties when deciding which cycle to protect.
On the other hand, the algorithm should not repeatedly break cycles and forgo immediate gain too frequently.
To balance immediate gain against potential future gain, we introduce a \emph{cool-down mechanism}.
Intuitively, depending on how long this break delays the burning of the remaining subgraph, the algorithm switches to the greedy choice whenever a firefighter becomes available during this period.
Once a greedy choice is made or the timer expires, the algorithm is again free to break a cycle.

\subsection{Analysis framework}

\paragraph*{Notion of analysis}

We first introduce critical components of our analysis. 
First, we narrow down the optimal solutions we are comparing against to a structured set, \emph{non-redundant} optimal solutions:

\begin{lemma}
    \label{lem:non-redundant}
    There exists a non-redundant optimal solution, where every cycle contains at most two protected vertices.
\end{lemma}

\begin{proof}
    Consider an optimal solution $\opt$ where there is a cycle $C$ with three distinct protected vertices on $C$.
    Let $u_0$ be the vertex on $C$ that has the smallest shortest distance to $r$.
    That is, $u_0 = \arg_{u\in C}\min \{\dist(G, r, u)\}$.
    Label the vertices on $C$ in cyclic order as $C = (u_0, u_1, u_2, \cdots, u_p)$.
    Note that an algorithm can only protect an available vertex. 
    Assume that $u_i$, $u_j$, and $u_k$ are the three protected vertices on $C$ with $i \leq j \leq k$.
    We can remove the protection on $u_j$ without decreasing the optimal solution's performance while preserving the property that no cycle has more than two protected vertices. 
\end{proof}

Therefore, from now on, we compare our algorithms' performance against a non-redundant optimal solution.

\medskip

\runtitle{Time indexing.}
Let $\algset = (\valg_1, \valg_2, \cdots)$ be the set of vertices protected by our algorithm, where the vertex $\valg_t$ is the $t$-th vertex protected by the algorithm according to the algorithm description.
We say that $t$ is the \emph{time label} of the vertex $v$ if $v = \valg_t$.
In short, we say that $\valg_t$ is protected \emph{at time $t$}.
Note that two vertices can be protected by our algorithms ``simultaneously''\footnote{That is, they are protected in the same iteration in the later presented algorithms; Algorithm~\ref{Alg:almostTree} (line~\ref{line:almostTreeSimultaniously}) and Algorithm~\ref{Alg:cactus} (line~\ref{line:cactusSimultaniously}).} to save a whole cycle.
In this case, these two vertices naturally have two consecutive time labels.

We further extend this time label idea to the input graph $\Galg_t$, which is the graph from the algorithm's perspective right before the vertex $\valg_t$ is protected.
Note that if $\valg_{t}$ and $\valg_{t+1}$ are protected in the same round, $\Galg_{t+1}$ is a subgraph of $\Galg_t$ with $\CoveredSet(\valg_{t})$ removed (see the description of the algorithms for how these graphs change throughout the game).
Moreover, consider the case where both $\Galg_t$ and $\Galg_{t+1}$ have a root cycle corresponding to the same cycle in the original graph, and $\valg_t$ and $\valg_{t+1}$ are not protected in the same round. 
The root cycle in $\Galg_{t+1}$ is smaller than the one in $\Galg_t$, because of the graph reduction procedure.
We also denote $\CoveredSet_t(S)$ as the set of vertices covered by the set of vertices $S \subseteq V(\Galg_t)$ and $\weight_t(S) = |\CoveredSet_t(S)|$ as the number of vertices saved correspondingly in graph $\Galg_t$.

We apply the time label to the set $\optset$ of vertices protected by $\opt$ as well.
Formally, let $\valg_{t+1}, \valg_{t+2}, \cdots, \valg_{t+f_i}$ be the vertices protected by the algorithm in round $i$. 
We fix an arbitrary one-to-one mapping between $\valg_{t+1}, \cdots, \valg_{t+f_i}$ and the $f_i$ vertices protected by the optimal solution in the same round (note that the algorithm cannot finish the game later than the optimal solution).
We denote by $\vopt_t$ the vertex protected by the optimal solution, if it is the image of $\valg_t$ under the one-to-one mapping.
Thus, we also order the vertices protected by the optimal solution $\vopt_1, \vopt_2, \cdots$. 
We say that $\valg_t$ and the corresponding $\vopt_t$ are protected \emph{at the same time}.

\medskip

\runtitle{Cycle-respecting partition.}
Given the input graph $G$ in which any two cycles share at most one vertex, our analysis framework is based on partitioning the set $\optset$ of vertices protected by \opt in a way that respects the cycles. 
Formally, a partition $\partition = \{P_1, P_2, \cdots\}$ is \emph{cycle-respecting} if for every cycle $C$ in $G$, the vertices in $C \cap \optset$ belong to the same part $P_i \subseteq \partition$.

\begin{lemma}\label{lem:cycle-respecting}
    Let $G$ be a cactus graph with fire source $r$, and $\optset$ is the subset of vertices protected by \opt.
    Suppose $\partition = \{P_1, P_2, \cdots, P_m\}$ is a partition of $\optset$ that is cycle-respecting.
    Then, $\CoveredSet(\optset) = \bigcup_{i=1}^m \CoveredSet(P_i)$.
\end{lemma}

\begin{proof}
    In order to show equality, it suffices to show mutual inclusion.
    We will first show $\bigcup_{i=1}^m \CoveredSet(P_i) \subseteq \CoveredSet(\optset)$.
    By definition of covered set, for any vertex $v \in \CoveredSet(P_i)$, every path from $v$ to $r$ intersects $P_i \in \optset$. 
    Thus, $\CoveredSet(P_i) \subseteq \CoveredSet(\optset)$.
    Therefore, $\bigcup_{i=1}^m \CoveredSet(P_i) \subseteq \CoveredSet(\optset)$.

    It remains to show the reverse inclusion $\CoveredSet(\optset) \subseteq \bigcup_{i=1}^m \CoveredSet(P_i)$.
    Consider $v \in \CoveredSet(\optset)$ and an arbitrary path $R_{v,r}$ from $v$ to $r$.
    By definition, $R_{v,r}$ intersects $\optset$. 
    Let $S^* \subseteq \optset$ be the set of vertices protected by \opt on $R_{v,r}$.
    If $S^*$ contains a vertex $\vopt$ that is not a cycle vertex, then $\vopt$ lies on every path from $v$ to $r$, implying $v \in \CoveredSet(\vopt)$.
    Therefore, $v \in \CoveredSet(P_i)$ for the set $P_i$ containing $\vopt$.

    Otherwise, suppose that every vertex in $S^*$ is a cycle vertex.
    Let $C$ be the first cycle intersecting $R_{v,r}$ (when traversing from $v$ to $r$) that contains a vertex of $S^*$.
    For any $\vopt \in S^* \cap C$, we must have $v \notin \CoveredSet(S^* \setminus {\vopt})$, since otherwise \opt would not be a non-redundant optimal solution.
    Let $v_i$ and $v_o$ denote the first and last vertices of $C \cap R_{v,r}$, respectively.
    Because $G$ is a cactus graph, every path from $v$ to $r$ intersects $C$ exactly at $v_i$ and $v_o$ (otherwise, an edge would belong to two cycles).
    If $v_i \in S^*$ (resp. $v_o \in S^*$), similar to the argument for non-cycle vertices, $v\in \CoveredSet(P_i)$ where $P_i$ contains $v_i$ (resp. $v_o$).   We therefore assume that neither $v_i$ nor $v_o$ belongs to $S^*$.
    
    Since $v\in \CoveredSet(\optset)$, both arcs of $C$ between $v_i$ and $v_o$ contain at least one vertex of $\optset$.
    By the cycle-respecting property, those two vertices lie in the same part $P_i \in \partition$, thus $v\in \CoveredSet(P_i)$.
    Therefore, for any $v \in \CoveredSet(\optset)$, $v\in \bigcup_{i = 1}^m \CoveredSet(P_i)$, and $\CoveredSet(\optset) \subseteq \bigcup_{i=1}^m \CoveredSet(P_i)$.
\end{proof}

\medskip

\runtitle{Exclusive covered set.}
For every subset $S^*\subseteq \optset$, we define its \emph{exclusive covered set} $\exclusiveCoveredSet(S^*)$ as $\CoveredSet(S^*)\setminus \CoveredSet(\algset)$.
That is, $\exclusiveCoveredSet(S^*)$ is the set of vertices covered by $S^*$ but not covered by the set of vertices $\algset$ protected by our algorithm.
Equivalently, $\exclusiveCoveredSet(S^*)$ can be considered as the covered set of $S^*$ on the graph $G\setminus\CoveredSet(\algset)$.
Accordingly, we define $\exclusiveWeight(S^*)$ as $|\exclusiveCoveredSet(S^*)|$.

\begin{lemma}[Exclusive weight as a lower bound]\label{lem:exclusiveWeight}
    For any $S^* \subseteq \optset$ and any $t \in \mathbb{N}$, $\exclusiveWeight(S^*) \leq \weight_t(S^*)$.
\end{lemma}

\begin{proof}
Since $\CoveredSet_t(S^*)$ includes all vertices in $\CoveredSet(S^*)$ that have not saved by $\algAlmost$ at time $t$, while $\exclusiveCoveredSet(S^*)$ includes vertices that are covered by $S^*$ but not ever by $\algAlmost$. 
Thus, $\exclusiveCoveredSet(S^*) \subseteq \CoveredSet_t(S^*)$, and $\exclusiveWeight(S^*) \leq \weight_t(S^*)$.    
\end{proof}

\medskip

\paragraph*{Framework of analysis: charging}

In the analysis, we partition the set $\optset$ of vertices protected by \opt while respecting the cycles. 
Then, we charge the covered set of each of the parts to the covered set of a subset of the protected set $\algset$ by the algorithm.
Formally, we define a cycle-respecting partition $\partition = (P_1, P_2, \cdots)$, which differs from the optimal solution on different graph classes. 
By Lemma~\ref{lem:cycle-respecting}, the optimal solution's profit $\opt(\instance) = \weight(\optset) = \sum_i\weight(P_i)$ for all $P_i \in \partition$.
Then, we split $\CoveredSet(P_i)$ into \emph{exclusive covered set} of $P_i$, $\exclusiveCoveredSet(P_i) = \CoveredSet(P_i)\setminus\CoveredSet(\algset)$ and $\CoveredSet(P_i)\cap\CoveredSet(\algset)$.
We further describe the profit of the optimal solution as $\opt(\instance) = \sum_i\weight(P_i) = \sum_i(\weight(P_i\cap \algset) + \exclusiveWeight(P_i)) \leq \weight(\algset) + \sum_i \exclusiveCoveredSet(P_i)$.

Then, we focus on $\exclusiveWeight(P_i)$ for a $P_i \in \partition$. 
We define a charging function that maps each $P_i$ to a subset of vertices protected by the algorithm based on the time labeling.
Formally, we define the charging function $\charging: \partition \to 2^{\algset}$.
Then, we show that $\exclusiveWeight(P_i) \leq \alpha_i \sqrt{n} \cdot \weight(\charging(P_i))$ for some constant $\alpha_i$ that may differ across parts $P_i$.
Let $\beta = \max_{\valg \in \algset} | \{P_i \in \partition \mid \{\valg \in \charging(P_i)\}|$ be the maximum number of times a vertex $\valg \in \algset$ is charged by parts in $\partition$.
Then, $\opt(\instance) = \sum_i \exclusiveWeight(P_i) \leq \sqrt{n} \cdot \sum_i \alpha_i \cdot \weight(\charging(P_i)) \leq \sqrt{n} \cdot \alpha_{\max} \cdot \beta \cdot \weight(\algset) = \sqrt{n} \cdot \alpha_{\max} \cdot \beta \cdot \alg(\instance)$, where $\alpha_{\max} = \max_i\{a_i\}$.
Namely, if we show that $\alpha_i$ is a constant for every $P_i \in \partition$ without charging any vertex $\valg \in \algset$ more than a constant number of times, then the algorithm is $O(\sqrt{n})$-competitive.

\section{Warm-up: $\mathbf{1}$-almost-tree}
\label{Sec:almostTree}
Recall that a $1$-almost-tree is a connected graph with at most one cycle, which is a superset of tadpole graphs, and any deterministic online algorithm is at least $\Omega(\sqrt{n})$-competitive. 
We first start with a warm-up on providing an $O(\sqrt{n})$-competitive algorithm on $1$-almost trees.

\subsection{$\algAlmost$: The algorithm for $1$-almost-trees}
Recall that by the reduction procedure, we focus on a single round $i$, of the firefighting game with a graph $\newG$ with a fire source $\newR$ and $f_i$ available firefighters. 

In round $i$, the algorithm $\algAlmost$ iteratively assigns firefighters to one of the vertices in $N(\newR)\cup\rootC$ until there are no firefighters left. 
In each iteration, if $\newR$ is not involved in any cycle (that is, there is no root cycle), $\algAlmost$ protects the heaviest neighbor of $\newR$.
When $\newG$ has a root cycle, $\algAlmost$ makes decisions between protecting the heaviest neighbor of $\newR$, protecting the root cycle with two firefighters, or breaking the root cycle.
At the end of one iteration, the protected vertices $v$ and $\CoveredSet(v)$ are removed from $\newG$, and the number of available firefighters is reduced accordingly.

\medskip

\runtitle{The fire source $\newR$ is not in any cycle.}
The algorithm $\algAlmost$ makes its decisions greedily if there is no root cycle in the input graph $\newG$.
That is, it assigns the firefighters to the first $f_i$ ``heaviest'' neighbors of $\newR$ that have the highest weights.
Formally, the firefighters are assigned iteratively to unprotected $v\in N(\newR)$ with the highest $\weight(v)$. 

\runtitle{The fire source $\newR$ is in a root cycle $\rootC$, and there are at least $2$ firefighters.}
In the case where the input graph $\newG$ has a root cycle $\rootC$, we consider the unprotected vertices in $N(\newR) \cup \rootC = \{v_1, v_2, \cdots\}$ where $\weight(v_1) \geq \weight(v_2) \geq \cdots$. 
When there are at least two firefighters available, we protect the cycle using two firefighters if $\weight(\rootC) > \weight(v_1)+\weight(v_2)$.
Otherwise, we protect $v_1$ using one firefighter, and the procedure iteratively continues (with $v_1$ removed from the candidates of the vertices that need to be protected).

\runtitle{The fire source $\newR$ is in a root cycle $\rootC$, and there is only one firefighter.}
When there is a root cycle $\rootC$ in the input graph $\newG$ and there is only one firefighter, the algorithm protects $v_1$ if $\weight(v_1) \geq \sqrt{\weight(\rootC)}$.
Otherwise, when the immediate gain from protecting $v_1$ is not significant enough, the algorithm ``breaks'' the cycle by protecting the vertex $u \in N(\newR) \cap \rootC$  with the largest $\tol(u, \rootC, \sqrt{\weight(\rootC)})$ (tolerance with respect to $\sqrt{\weight(\rootC)}$).
Recall that the tolerance $\tol(u, \rootC, \sqrt{\weight(\rootC)})$ of a cycle vertex $u \in \rootC$ is the maximum distance that the fire can spread while keeping $\sqrt{\weight(\rootC)}$ vertices unburned if $u$ is protected, thereby breaking the cycle $\rootC$.
Intuitively, the ``cycle-break'' procedure protects the cycle vertex that maximizes how far the fire can spread while still leaving at least $\sqrt{\weight(\rootC)}$ vertices unburned (Lemma~\ref{lem:NeighborsAreTheBest}).
Algorithm~\ref{Alg:almostTree} provides the detailed description of the algorithm for $1$-almost-trees.

\medskip

\begin{algorithm}[t]
    \caption{Procedure of $\algAlmost$ for $1$-almost tree in round $i$} 
    \label{Alg:almostTree}
    \textbf{Input:} Graph $\newG$ with fire source $\newR$ and $f_i$ available firefighters\\
    \While{$f_i > 0$ and there are available vertices in $\newG$}{
        \eIf{$\newR$ is not a cycle vertex in $\newG$}{
            Protect the vertex $v$ in $N(\newR)$ with the highest weight, \\
            Set $f_i\gets f_i -1$, $\newG \gets \newG \setminus \CoveredSet( v )$.
        }{
            \tcp{$\newR \in \rootC$ for some root cycle $\rootC$}
            Consider unprotected vertices $v_1, v_2, \cdots$ in $N(\newR)\cup \rootC$ with $\weight(v_1) \geq \weight(v_2) \geq \cdots$

            \eIf{$f_i \geq 2$}{
                \eIf{$\weight(v_1) + \weight(v_2) \geq \weight(\rootC)$}{   
                    Protect $v_1$.\\ 
                    Set $f_i\gets f_i -1$, $\newG \gets \newG \setminus \CoveredSet( v_1 )$.
                }{
                    Protect both vertices $v_\ell, v_r$ in $N(\newR) \cap \rootC$.\label{line:almostTreeSimultaniously}\\
                    Set $f_i\gets f_i -2$, $\newG \gets \newG \setminus \CoveredSet( \{v_\ell, v_r\} )$.
                }
            }{
                \eIf{$\weight(v_1) \geq \sqrt{\weight(\rootC)}$}{
                        Protect $v_1$.\\
                        Set $f_i\gets f_i -1$, $\newG \gets \newG \setminus \CoveredSet(v_1 )$.
                }{
                    \texttt{Break-cycle:} Protect the vertex $u \in N(\newR) \cap \rootC$ with the largest tolerance $\tol(u, \rootC, \sqrt{\weight(\rootC)})$.\\
                    Set $f_i\gets f_i -1$, $\newG \gets \newG \setminus \CoveredSet( u )$.
                }
                
            }
        }
    }
\end{algorithm}

\begin{observation}[Properties of $\algAlmost$]
    \label{obs:algAlmost}
    When the input graph has at most one cycle, by the definition of $\algAlmost$:
    \begin{enumerate}[(a)]
        \item \label{obs:algAlmostMax} 
        If a vertex $v^\texttt{A}$ is chosen at time $t$ in $\Galg_t$ because it is the vertex with the largest weight, then $\weight_t(v^\texttt{A})^2 \geq \weight_t(\rootC)$.
        
        \item \label{obs:algAlmostCycle}
        If a vertex $v^\texttt{A}$ is chosen at time $t$ in $\Galg_t$ because it is chosen---together with another vertex---to save the cycle, then $\weight_t(\rootC) > \weight_t(x)+\weight_t(y)$ for all vertices $x$ and $y$ in $\Galg_t$. 
        Moreover, $\weight_t(\rootC) \geq \weight_t(\{x, y\})$ (which can be observed by a case distinction on whether each of $x$ and $y$ is on the cycle).

        \item \label{obs:algAlmostBreakCycle}
        If a vertex $v^\texttt{A}$ is chosen at time $t$ in $\Galg_t$ because it is selected as the cycle-breaking point, then for every vertex $v$ in $\Galg_t$, $\weight_t(v) \leq \weight_t(\tHeaviest) < \sqrt{\weight_t(\rootC)}$, where $\tHeaviest$ is the vertex with the heaviest weight in $N(\newR) \cup \rootC$ in $\Galg_t$.
    \end{enumerate}
\end{observation}

\paragraph*{Structural properties of \texttt{Break-cycle}}
We first show that for $1$-almost trees, it is ``beneficial'' to break a cycle by protecting a cycle vertex adjacent to the fire source.
We state the following lemma in a more general form than is presently required, as it will be useful later when we consider cactus graphs.

\begin{figure}[bt]
    \centering
    \includegraphics[width=0.75\linewidth]{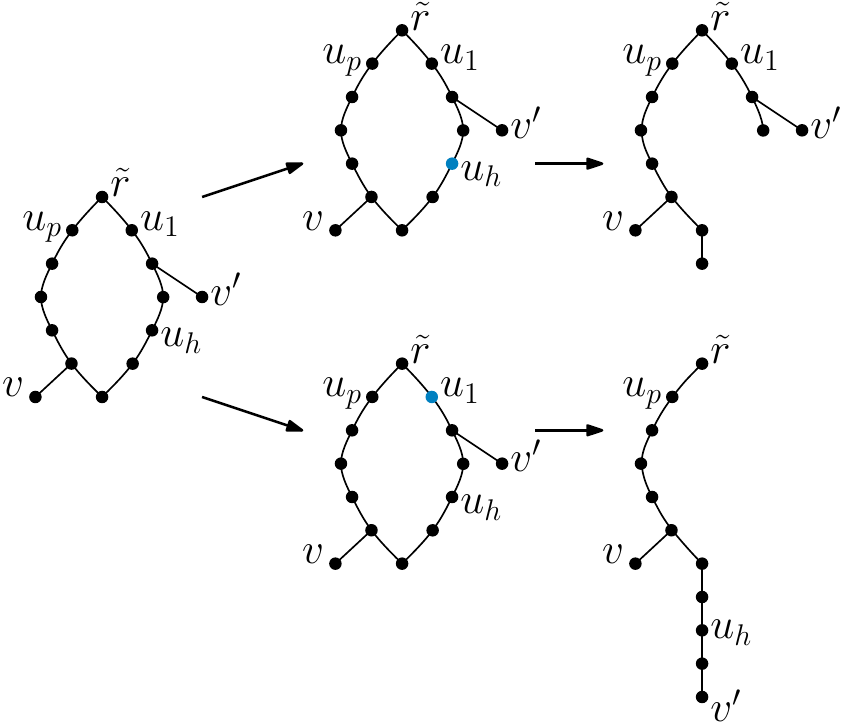}
    \caption{The distances from $v, v^\prime$ to $\newR$ in $\newT{\rootC}{u_h}$ compared to $\newT{\rootC}{u_1}$.}
    \label{fig:lemma9}
\end{figure}

\begin{lemma}[The advantage of breaking cycle on $N(\newR) \cap C$]
    \label{lem:NeighborsAreTheBest}    
    In a cactus graph $\newG$ with a fire source $\newR$ and a root cycle $\rootC = (r, u_1, \cdots, u_p)$, for every $u_h$, one of the following holds for every $d\in \mathbb{N}$:
    \begin{align*}
    &\amountBelow(\newT{\rootC}{u_h}, d) \leq \amountBelow(\newT{\rootC}{u_1},d) + \weight(u_1) \text{, or}\\
    &\amountBelow(\newT{\rootC}{u_h}, d) \leq \amountBelow(\newT{\rootC}{u_p},d) + \weight(u_p).
    \end{align*}
\end{lemma}

\begin{proof}
Consider an arbitrary $u_h$ and $d \in \mathbb{N}$. 
Without loss of generality, assume that the shortest path from $u_h$ to $\newR$ passes through $u_1$. 

Consider an arbitrary vertex $v \notin \CoveredSet(u_1)$ that is counted in $\amountBelow(\newT{\rootC}{u_h}, d)$, we first show that it must be also counted in $\amountBelow(\newT{\rootC}{u_1}, d)$.
By the fact that $v$ is counted in $\amountBelow(\newT{\rootC}{u_h}, d)$, the shortest distance from $v$ to $\newR$ in $\newT{\rootC}{u_h}$, $\dist(\newT{\rootC}{u_h}, v, \newR) \geq d$.
Let $P_{v\newR}$ and $P'_{v\newR}$ denote the shortest paths from $v$ to $\newR$ in $\newT{\rootC}{u_h}$ and $\newT{\rootC}{u_1}$, respectively.
Note that $\newT{\rootC}{u_1}$ is a path from $\newR$ to $u_2$ union all $\CoveredSet(u_k)$ for $k = 2, 3, \cdots, p$. 

Consider the shortest paths from $v$ to $\newR$ before and after removing $u_h$. 
Either removing $u_h$ change the shortest path or not.
We first consider the case where $u_1$ is on the path $P_{v\newR}$.
This implies, as we have assumed that the shortest path from $u_h$ to $\newR$ passes through $u_1$, that $P_{v\newR}$ is the shortest path from $v$ to $\newR$ in $\newG$.
Then, by definition of a shortest path and the fact that the deletion of a vertex cannot shorten paths, $P^\prime_{v\newR}$ is at least as long as $P_{v\newR}$.
And thus, $\dist(\newT{\rootC}{u_1}, v, \newR) \geq \dist(\newT{\rootC}{u_h}, v, \newR) \geq d$. 
Therefore, $v$ is also counted in $\amountBelow(\newT{\rootC}{u_1}, d)$.

Second, we consider the remaining case where $u_1$ is not on the path $P_{v\newR}$.
Then, the the path $P_{v\newR}$ would not be affected by the deletion of vertex $u_1$.
Hence, the paths $P_{v\newR}$ and $P^\prime_{v\newR}$ are identical.
And thus, $\dist(\newT{\rootC}{u_1}, v, \newR) = \dist(\newT{\rootC}{u_h}, v, \newR) \geq d$. 
Therefore, $v$ is also counted in $\amountBelow(\newT{\rootC}{u_1}, d)$.

Finally, the vertices in $\CoveredSet(u_1)$ might be counted in $\amountBelow(\newT{\rootC}{u_h})$ but never counted in $\amountBelow(\newT{\rootC}{u_1}$.
Since, by definition, $\newT{\rootC}{u_1}$ is the subgraph induced by the vertices in $\{\newR\} \cup \CoveredSet(\rootC) \setminus \CoveredSet(u_1)$.
Thus, it is proven that $\amountBelow(\newT{\rootC}{u_h}, d) \leq \amountBelow(\newT{\rootC}{u_1}, d) + \weight(u_1)$.

\medskip

The same argument applies when the shortest path from $u_h$ to $\newR$ passes through $u_p$, in which case $\amountBelow(\newT{\rootC}{u_h}, d) \leq \amountBelow(\newT{\rootC}{u_p},d) + \weight(u_p)$.
\end{proof}

Next, we bound the performance of breaking the cycle on the vertex $\algAlmost$ picks in terms of how many vertices can have a safe distance from the fire source due to the breaking.

\begin{lemma}[Quality of \texttt{Break-cycle}]
    \label{lem:CycleBreakIsGood}
    Let $\newG$ be the input graph with a root cycle $\rootC = (r, u_1, u_2, \cdots, u_p)$. 
    Let $\hat{u}$ be the vertex protected by the \texttt{Break-cycle} procedure. 
    For every vertex $u_h$ and every $d\in \mathbb{N}$, $\frac{\amountBelow(\newT{\rootC}{u_h}, d)}{\amountBelow(\newT{\rootC}{\hat{u}}, d)+1} \leq 2\sqrt{\weight(\rootC)}$.
\end{lemma}

\begin{proof}
    Let $\rootC= (r, u_1, u_2, \cdots, u_p)$ be the vertices on $\rootC$ with $u_1=\hat{u}$.
    Denote the tolerance of $\hat{u}$ by $\dmax = \tol(\hat{u}, \rootC, \sqrt{\weight(\rootC)})$.
    By the definition of tolerance, $\amountBelow(\newT{\rootC}{\hat{u}}, \dmax) \geq \sqrt{\weight(\rootC)}$.

    \medskip
    
    \runtitle{Case $1$: $d \leq \dmax$.}
    By the monotonicity of the $\amountBelow$ function, $\amountBelow(\newT{\rootC}{\hat{u}}, d) \geq \amountBelow(\newT{\rootC}{\hat{u}}, \dmax) \geq \sqrt{\weight(\rootC)}$.
    On the other hand, for any $u_h$, $\amountBelow(\newT{\rootC}{u_h}, d)$ can include at most every vertex in $\CoveredSet(\rootC)$. 
    Thus, $\amountBelow(\newT{\rootC}{u_h}, d) \leq \weight(\rootC)$.
    Therefore, $\frac{\amountBelow(\newT{\rootC}{u_h}, d)}{\amountBelow(\newT{\rootC}{u_1}, d)+1} \leq \frac{\weight(\rootC)}{\sqrt{\weight(\rootC)}+1} \leq 2\sqrt{\weight(\rootC)}$.
    
    \medskip
    
    \runtitle{Case $2$: $d > \dmax$.}
    Since $\dmax$ is the tolerance $\tol(\hat{u}, \rootC, \sqrt{\weight(\rootC)})$, $\dmax$ is the maximum distance the fire can spread $\newR$ in $\newT{\rootC}{\hat{u}}$ such that there are still at least $\sqrt{\weight(\rootC)}$ vertices unburned.
    Thus, for any $d > \dmax$, $\amountBelow(\newT{\rootC}{\hat{u}}, d) < \sqrt{\weight(\rootC)}$.
    Since the algorithm selected $\hat{u}$ out of $u_1$ and $u_p$ based on which neighbor had the highest tolerance, $\tol(\newT{\rootC}{u_1} \leq \tol(\newT{\rootC}{\hat{u}}$ and $\tol(\newT{\rootC}{u_p} \leq \tol(\newT{\rootC}{\hat{u}}$.
    Moreover, it follows that $\amountBelow(\newT{\rootC}{u_1}, d) < \sqrt{\weight(\rootC)}$ and $\amountBelow(\newT{\rootC}{u_p}, d) < \sqrt{\weight(\rootC)}$.
    Additionally, by $\algAlmost$ and Observation~\ref{obs:algAlmost}(\ref{obs:algAlmostBreakCycle}), the procedure \texttt{Break-cycle} only happens when $\weight(v) \leq \sqrt{\weight(\rootC)}$ for all vertices $v \in V$.
    Therefore, by Lemma~\ref{lem:NeighborsAreTheBest}, $\amountBelow(\newT{\rootC}{u_h}, d) \leq \max\{\amountBelow(\newT{\rootC}{u_1}, d) + \weight(u_1), \amountBelow(\newT{\rootC}{u_p}, d) + \weight(u_p)\} < \sqrt{\weight(\rootC)} + \sqrt{\weight(\rootC)} = 2\sqrt{\weight(\rootC)}$.   
\end{proof}

\subsection{Competitiveness analysis of $\algAlmost$}

\paragraph*{Charging framework}
Recall that from the time labeling scheme, each protected vertex $v$ in $\algset$ has a distinct and contiguous integral time label $t_v$. 
Moreover, for each vertex $v\in \algset$, there is a protected vertex in $\optset$ with the same time label. 
We denote this corresponding vertex in $\optset$ by $v^*$.

\medskip

\runtitle{Partition $\mathcal{P}^*$ of $\optset$.}
We partition $\optset$ into the three parts:
\begin{enumerate}[$P^a$:]
    \item \label{AlmostPart1} \runtitle{Cycle vertices.}
        All vertices in $\optset$ that are cycle vertices in the original graph $G$,
    \item \label{AlmostPart2} \runtitle{Non-cycle vertices and available to $\algAlmost$.}
        All vertices $\vopt_t \in \optset$ that are available (unburned) in $\Galg_t$, and
    \item \label{AlmostPart3} \runtitle{Non-cycle vertices and not available to $\algAlmost$.}
        All vertices $\vopt_t \in \optset$ that are burned or saved in $\Galg_t$.
\end{enumerate}

\begin{observation}
    The partition $\mathcal{P}^*$ is a proper, cycle-respecting partition of $\optset$.
\end{observation}

\begin{proof}
    Every vertex in $\optset$ is either a cycle vertex or a non-cycle vertex. 
    Additionally, every non-cycle vertex $\vopt_t \in \optset$ either is available in $\Galg_t$ or is unavailable in $\Galg_t$. 
    Thus, every vertex in $\optset$ belongs to exactly one of the three parts, thereby forming a proper partition of $\optset$. 
    Moreover, since all cycle vertices in $\optset$ are included in $P^{\ref{AlmostPart1}}$, the partition is cycle-respecting.
\end{proof}

\medskip

\runtitle{Charging function.}
For the analysis, we define the charging function $\charging: \mathcal{P}^* \to 2^{\algset}$ that charges each part of $\mathcal{P}^*$ to a subset of vertices protected by $\algAlmost$. 
In the case of $\algAlmost$, we simply charge every part $P^{\boldsymbol{\cdot}}\in \mathcal{P}^*$ to the entire set $\algset$ by defining $\charging(P^{\boldsymbol{\cdot}}) = \algset$.
That is, we charge each part to $\algset$ individually.
Note that since we charge each part to $\algset$ directly, the analysis framework is a bit easier than the general one.

In the following, we focus on each part $P^{\boldsymbol{\cdot}} \in\mathcal{P}^*$ and show that the size of the exclusive covered set of $P^{\boldsymbol{\cdot}}$ is at most $O(\sqrt{n})$ times the size of the covered set of $\algset$, where $n$ is the total number of vertices in the original graph $G$.

\begin{lemma}[Cycle vertices in $\optset$]\label{lem:algAlmostPa}
    Given the part $P^a$ in $\mathcal{P^*}$, $\exclusiveWeight(P^a) \leq \sqrt{n} \cdot \weight(\algset)$, where $n$ is the number of vertices in the input graph $G$.
\end{lemma}

\begin{proof}
    Recall that by Lemma~\ref{lem:non-redundant}, we focus on a non-redundant optimal solution $\opt$. 
    Thus, there are at most two vertices in $P^a$.
    If $P^a = \emptyset$, the inequality holds.
    So we assume that there is at least one vertex in $P^a$.

    Let $t$ be the time (label) when \opt protects the first cycle vertex $\vopt_t$ in $P^a$. 
    We consider the graph $\Galg_t$ (the graph right before the corresponding vertex $\valg_t$ was protected by $\algAlmost$ from the algorithm's perspective), where the heaviest vertex in $N(\newR) \cup \rootC$ is $\tHeaviest$ (note that $\tHeaviest$ can also be a cycle vertex). 

    \runtitle{$\Galg_t$ contains no cycle.}
    Consider the case where $\Galg_t$ contains no cycle.
    This can occur in one of two ways: either the original cycle $C$ has already been covered by protecting an ancestor vertex before time~$t$, when $C$ was not yet a root cycle in the graph of \opt, or the cycle has been broken before time~$t$ by protecting the vertex $u_c$ on the original cycle~$C$.
    
    For the former, the algorithm has saved the entire cycle, thus $\CoveredSet(C) \subseteq \CoveredSet(\algset)$.
    Then, as $\CoveredSet(P^a) \subseteq \CoveredSet(C)$, the exclusive covered set $\exclusiveCoveredSet(P^a) = \CoveredSet(P^a) \setminus \CoveredSet(\algset) = \emptyset$.
    Thus, $\exclusiveWeight(P^a) = 0$ and the inequality holds.
    
    For the latter, by the algorithm's definition, it protects $\tHeaviest$ at time $t$ (namely, $\valg_t = \tHeaviest$)
    Since the protection on $u_c$ has broken the cycle $C$ into two arcs $A_1$ and $A_2$ by time $t$, and $\tHeaviest$ is the heaviest vertex at time $t$, $\weight(\algset) \geq \weight(\{u_c, \tHeaviest\}) \geq \weight_t(\{u_c, \tHeaviest\}) \geq \max\{\weight_t(\{ u_c \} \cup A_1), \weight_t( \{ u_c \} \cup A_2)\} \geq \weight_t(P^a)/2 \geq \exclusiveWeight(P^a)/2$.
    Where the second-last inequality holds because the remaining (unburned) vertices of the arcs $A_1, A_2$, together with vertex $u_c$ (if not already burned) make up the remaining cycle in the graph of $\opt$ (as the fire spreads with the same speed in both graphs), and $\CoveredSet(P^a)$ is at most as large as the covered set of the remaining cycle in the graph of \opt.
    Thus, the lemma holds in this case.

    \runtitle{$\Galg_t$ contains a root cycle.}
    In contrast, consider the case where $\Galg_t$ still has a root cycle $\rootC$. 
    Thus, $\exclusiveWeight(P^a) \leq \weight_t(\rootC)$.
    At time $t$, $\algAlmost$ selects $\valg_t$ in $\Galg_t$ by making one of the following three decisions: 1) protecting the heaviest vertex in $\Galg_t$, 2) protecting $\rootC$ by protecting two cycle vertices, or 3) breaking cycle $\rootC$. 
    In the following, we show that $\weight(\algset) \geq \sqrt{\weight_t(\rootC)}$ in the first two cases.
    Thus, since $\exclusiveWeight(P^a) \leq \weight(P^a) \leq \weight_t(\rootC)$, the lemma holds.

    If $\valg_t$ is selected as the heaviest vertex in $\Galg_t$, by Observation~\ref{obs:algAlmost}(\ref{obs:algAlmostMax}), $\weight_t(\valg_t)^2 \geq \weight_t(\rootC)$.
    Thus, $\weight(\algset) \geq \weight_t(\valg_t) \geq \sqrt{\weight_t(\rootC)}$.
    If $\valg_t$ is selected to protect $\rootC$ with another cycle vertex $u$, then $\weight(\algset) \geq \weight(\{\valg_t, u\}) \geq \weight_t(\{\valg_t, u\}) = \weight_t(\rootC) \geq \sqrt{\weight_t(\rootC)}$. 

    \medskip
    
    The third case, where $\valg_t$ is chosen to break the root cycle $\rootC$, is more involved. 
    In this case, $\weight_t(v) \leq \weight_t(\tHeaviest) < \sqrt{\weight_t(\rootC)}$ for all $v \in \Galg_t$, where $\tHeaviest$ is the heaviest vertex in $\Galg_t$ (by Observation~\ref{obs:algAlmost}(\ref{obs:algAlmostBreakCycle})).
    Then, if $|P^a| = 1$, $\exclusiveWeight(P^a) \leq \sqrt{\weight_t(\rootC)} \leq \sqrt{\weight_t(\rootC)} \cdot \weight(\algset)$, noting that $\weight(\algset) \geq \weight(\valg_t) \geq 1$.
    
    If instead $|P^a| = 2$, let $t$ and $s = t + \delta$ be the times at which \opt protects the two cycle vertices in $C$, where $\delta \geq 1$. 
    In the $\delta$ iterations between time $t$ and time $s$, the fire spreads to at most $\delta$ additional vertices on $\rootC$. 
    Thus, the segment on $\rootC$ saved by \opt has length at most $|\rootC| - 1 - \delta$. 
    Recalling that $\weight_t(v)\leq \sqrt{\weight_t(\rootC)}$ at time $t$, $\exclusiveWeight(P^a) \leq \weight_t(P^a) \leq (|\rootC| - 1 - \delta)\cdot \sqrt{\weight_t(\rootC)}$.
    Meanwhile, the algorithm protects $\valg_t$ to break $\rootC$ close to the root, and when another firefighter is available at time $s$, $\algAlmost$ always protects the heaviest vertex in $\Galg_s$. 
    Thus, $\weight(\algset) \geq \weight(\{\valg_t, \valg_s\}) \geq \weight_t(\{\valg_t, \valg_s\}) \geq |\rootC| - 1 - \delta$. 
    Therefore, $\exclusiveWeight(P^a) \leq \sqrt{\weight_t(\rootC)}\cdot \weight(\algset)$ in this case.
\end{proof}

\medskip

\begin{lemma}[Non-cycle vertices in $\optset$ that are available to $\algAlmost$]\label{lem:algAlmostPb}
    Consider the part $P^b$ in $\mathcal{P^*}$, which consists of all non-cycle vertices $\vopt_t \in \optset$ that are available in $\Galg_t$. 
    Then $\exclusiveWeight(P^b) \leq \sqrt{n} \cdot \weight(\algset)$, where $n$ is the number of vertices in the input graph $G$.
\end{lemma}

\begin{proof}
    Consider $\vopt_t \in P^b$ and the corresponding vertex $\valg_t$ protected by $\algAlmost$ at the same time $t$. 
    By the definition of $P^b$, at least $\vopt_t$ is still available for $\algAlmost$.
    Therefore, $\Galg_t$ is not empty.
    We denote $\tHeaviest$ as the heaviest vertex in $\Galg$ and $\rootC$ as the root cycle (if it exists).

    There are three possibilities of $\valg_t$: 1) $\valg_t$ is protected as the heaviest vertex in $\Galg_t$, 2) $\valg_t$ is protecting $\rootC$ with another cycle vertex, or 3) $\valg_t$ is protected to break cycle $\rootC$.
    We further partition $P^b$ into $P^b_1, P^b_2$ and $P^b_3$ according to which case the corresponding vertex $\valg_t$ is protected in.
    
    In the first case, by Observation~\ref{obs:algAlmost}(\ref{obs:algAlmostMax}), $\weight(\valg_t) \geq \weight_t(\valg_t) \geq \weight_t(\vopt_t) \geq \exclusiveWeight(\vopt_t)$.
    Then, $\exclusiveWeight(P^b_1) = \sum_{t:\vopt_t\in P^b_1} \exclusiveWeight(\vopt_t) \leq \sum_{t:\vopt_t\in P^b_1} \weight(\valg_t)$.
    
    In the second case, $\valg_{t}$ protects the cycle $\rootC$ together with another cycle vertex~$\valg_k$ where $k\in \{t-1, t+1 \}$.
    By Observation~\ref{obs:algAlmost}(\ref{obs:algAlmostCycle}), $\weight(\{\valg_t, \valg_{k}\}) \geq \weight_t(\{\valg_t, \valg_{k}\}) = \weight_t(\rootC) > \weight(\vopt_t) + \weight(\vopt_{k}) \geq \exclusiveWeight(\vopt_t) + \exclusiveWeight(\vopt_k)$.
    Note that if $\vopt_k \in P^b$, then $\vopt_k$ must be in $P^b_2$, since at time $k$, $\algAlmost$ protected $\rootC$ together with another cycle vertex.
    Therefore no other vertex in $P^b \setminus P^b_2$ will be mapped to $\valg_k$, and $\exclusiveWeight(P^b_2) = \exclusiveWeight(\vopt_t) + \exclusiveWeight(\vopt_k) \leq \weight(\{\valg_t, \valg_k \})$.
    Otherwise, when $\vopt_k \notin P^b$, since $\rootC$ is the unique cycle in the graph, it follows that $\exclusiveWeight(P^b_2) = \exclusiveWeight(\vopt_t) \leq \exclusiveWeight(\vopt_t) + \exclusiveWeight(\vopt_k) \leq \weight(\{\valg_t, \valg_k \})$.
    
    In the third case, by Observation~\ref{obs:algAlmost}(\ref{obs:algAlmostBreakCycle}), $\exclusiveWeight(\vopt_t) \leq \weight_t(\tHeaviest) \leq \sqrt{\weight_t(\rootC)}$.
    Since $\weight_t(\valg_t) \geq 1$, $\exclusiveWeight(\vopt_t) \leq \sqrt{\weight_t(\rootC)}\cdot \weight_t(\valg_t)$.
    Then, as $\sqrt{\weight_t(\rootC)} \leq \sqrt{n}$, it follows that $\exclusiveWeight(P^b_3) = \sum_{t:\vopt_t\in P^b_3} \exclusiveWeight(\vopt_t) \leq \sum_{t:\vopt_t\in P^b_3} \sqrt{n} \cdot \weight(\valg_t)$.

    Finally, letting $k$ and $k+1$ be the time steps where $\algAlmost$ protects the cycle, $\exclusiveWeight(P^b) = \exclusiveWeight(P^b_1) + \exclusiveWeight(P^b_2) + \exclusiveWeight(P^b_3) \leq \sum_{t:\vopt_t\in P^b_1} \weight(\valg_t) + \weight(\{\valg_k, \valg_{k+1} \}) + \sum_{t:\vopt_t\in P^b_3} \sqrt{n} \cdot \weight(\valg_t) \leq \sqrt{n} \cdot \exclusiveWeight(\algset)$.
    Where the last inequality holds because each vertex $\valg_t$ is mapped to by a single other vertex, specifically, $\vopt_t$.
    Hence, the images of each part, that is $\charging (P^b_{\boldsymbol{\cdot}} )$, are non-intersecting and the union of them only contains vertices protected by $\algAlmost$ and therefore form a subset of $\algset$.
\end{proof}

What remains are the non-cycle vertices in $\optset$ that are unavailable to $\algAlmost$.
There are two possibilities for a vertex $\vopt_t$ to be unavailable to $\algAlmost$.
Either the vertex is saved in a previous round by $\algAlmost$, or the vertex is burned.
The first case is covered by our charging scheme.

In the latter case, it is important to note that the vertex $\vopt_t$ is still available to the optimal algorithm.
Since the fire spreads at a constant speed in both the graph maintained by our algorithm and the graph maintained by the optimal algorithm, it must be the case that the optimal solution managed to delay the fire in reaching the vertex.
The only way to delay the fire, without altogether stopping the spread of the fire and thus saving the vertex, is to protect a cycle vertex on the shortest path from $\vopt_t$ to $\newR$.
The fire now is still able to reach vertex $\vopt_t$, but it reaches $\vopt_t$ at a later time than it would in a graph where no such cycle vertex is protected.
Hence, it must be the case that $\algAlmost$ did either not break the cycle, or it broke the cycle in a different way, i.e, it protected a cycle vertex that is not on the shortest path from $\vopt_t$ to $\newR$.
This is illustrated in Figure~\ref{fig:lemma14}.

\begin{figure}[bt]
    \centering
    \includegraphics[width=0.5\linewidth]{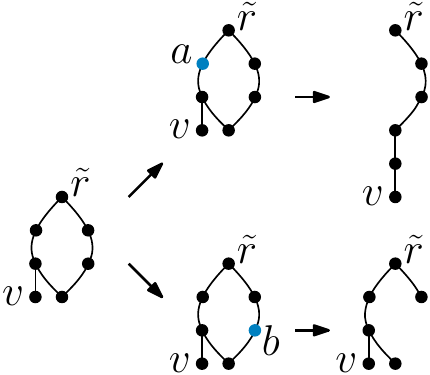}
    \caption{Delay.
    The fire can be delayed by protecting a cycle vertex ($a$) on the unique shortest path from a vertex to the root.
    Protecting a vertex $b$ on a not (unique) shortest path does not affect the distance to the root.}
    \label{fig:lemma14}
\end{figure}

\begin{lemma}[Non-cycle vertices in $\optset$ that are unavailable to $\algAlmost$]\label{lem:algAlmostPc}
    Consider the part $P^c$ in $\mathcal{P^*}$, which consists of all non-cycle vertices $\vopt_t \in \optset$ that are unavailable in $\Galg_t$. 
    Then $\exclusiveWeight(P^c) \leq (4\sqrt{n}+1) \cdot \weight(\algset)$, where $n$ is the number of vertices in the input graph $G$.
\end{lemma}

\begin{proof}
    First, consider that vertex $\vopt_t \in P^c$ is saved by $\algAlmost$ in an earlier round $s < t$.
    Then, it must be the case that $\vopt_t \in \CoveredSet(\valg_s)$, or $\vopt \in \CoveredSet( \{ \valg_s, \valg_{s+1} \})$ in the case where $\algAlmost$ protected a cycle in round $s$.
    Note that by the definition of the exclusive covered set of a vertex, it follows that $\exclusiveCoveredSet(\vopt_t) = \CoveredSet(\vopt_t) \setminus \CoveredSet(\algset) = \emptyset$.
    And thus, $\exclusiveWeight(\vopt_t) = 0$.    

    Observe that for each vertex that is unavailable for $\algAlmost$ because it was saved in a previous round, the exclusive weight of this vertex is $0$.
    Thus, let us focus on just the vertices that are unavailable because they are burned in $\Galg_t$.
    That is, let $P^c$ be the vertices that are burned in $\Galg_t$ but still available for $\opt$.
    Note that for the fire to be delayed, \opt must have protected a cycle vertex on the shortest path from $\vopt_t$ to the root at time $s < t$, while $\algAlmost$ did not break the cycle or broke the cycle not on the shortest path.
    Then, we consider the graph $\Galg_s$, which has a root cycle $\rootC$ from the algorithm's perspective. 
    Let $\sHeaviest$ be the heaviest vertex in $\Galg_s$. 
    Since $P^c$ are non-cycle vertices and $P^c \subseteq \CoveredSet_s(\rootC)$ in $\Galg_s$, $\exclusiveWeight(P^c) \leq \weight(P^c) \leq \weight_s(\rootC)$. 
    At time $s$, $\algAlmost$ assigns a firefighter to $\valg_s$ to 1) protect the heaviest vertex in $\Galg_s$, 2) protect $\rootC$ with another cycle vertex, or 3) break cycle $\rootC$.
    We further partition $P^c$ into $P^c_1$, $P^c_2$, and $P^c_3$ according to which case the corresponding vertex $\valg_s$ is protected.
    Note that since the vertices in $P^c$ are not available to the algorithm and the original input graph has only one cycle, the covered sets of $P^c_1$, $P^c_2$, and $P^c_3$ have no intersection.
    Thus, $\exclusiveWeight(P^c) = \exclusiveWeight(P^c_1) + \exclusiveWeight(P^c_2) + \exclusiveWeight(P^c_3)$.

    \runtitle{$\mathbf{P^c_1}$ and $\mathbf{P^c_2}$.}
    In the first case, by Observation~\ref{obs:algAlmost}(\ref{obs:algAlmostMax}), $\weight(\algset) \geq \weight_s(\valg_s) \geq \sqrt{\weight_s(\rootC)}$. 
    In the second case, by Observation~\ref{obs:algAlmost}(\ref{obs:algAlmostCycle}), $\weight(\algset) \geq \weight_s(\rootC)$.
    Thus, as $\exclusiveWeight(P^c) \leq \weight_s(\rootC)$, $\exclusiveWeight(P^c_1) \leq \exclusiveWeight(P^c) \leq \sqrt{\weight_s(\rootC)}\cdot \weight(\algset) \leq \sqrt{n} \cdot \weight(\algset)$ and $\exclusiveWeight(P^c_2) \leq \exclusiveWeight(P^c) \leq \weight(\algset)$.
    
    \runtitle{$\mathbf{P^c_3}$.}
    In the third case, since $\algAlmost$ also breaks the cycle at time $s$, the graphs $\algAlmost$ and $\opt$ play on are both trees. 
    By Observation~\ref{obs:algAlmost}(\ref{obs:algAlmostBreakCycle}), $\weight_s(v) \leq \sqrt{\weight_s(\rootC)}$ for every vertex $v$ in $\Galg_s$ (recall that $\Galg_s$ is the graph right before $\valg_s$ was protected).
    Therefore, since the cycle $\rootC$ (which is a superset of $P^c$) was in both $\Galg_s$ and $G^*_s$, in every following iteration, if \opt protects a vertex in $P^c_d$, the vertex weighs at most $\sqrt{\weight_s(\rootC)}$.
    Meanwhile, $\algAlmost$ either matches each of these vertices with a distinct one of its own or the game ends from $\algAlmost$'s point of view.
    Let $\delta$ be the distance the fire spreads between the time $s$ and the time the game finishes for $\algAlmost$.
    From this time onward, \opt can save at most $\amountBelow(\newT{\rootC}{\vopt_s}, \delta)$ additional vertices by protecting vertices in $P^c$.
    Meanwhile, since $\algAlmost$ finished the game by this time, $\algAlmost$ has covered at least $\amountBelow(\newT{\rootC}{\valg_s}, \delta) + \weight_s(\valg_s)$. 
    By Lemma~\ref{lem:CycleBreakIsGood}, $\frac{\amountBelow(\newT{\rootC}{\vopt_s}, \delta)}{\amountBelow(\newT{\rootC}{\valg_s, \delta)+ \weight_s(\valg_s)}} \leq 2\sqrt{\weight(\rootC)}$. 
    Thus, $\exclusiveWeight(P^c_3) \leq 2\sqrt{\weight(\rootC)} \cdot \weight(\algset)$.
    
    Altogether, $\exclusiveWeight(P^c) = \exclusiveWeight(P^c_1) + \exclusiveWeight(P^c_2) + \exclusiveWeight(P^c_3) \leq (3 \sqrt{\weight(\rootC)} + 1) \cdot \weight(\algset) \leq (3\sqrt{n} + 1) \cdot \weight(\algset)$.

\end{proof}

\begin{proof}[Proof of Theorem~\ref{thm:Almost}]
    By Lemmas~\ref{lem:algAlmostPa},~\ref{lem:algAlmostPb}, and~\ref{lem:algAlmostPc}, $\sum_{P^{(\boldsymbol{\cdot})}\in \partition} \exclusiveWeight(P^{(\boldsymbol{\cdot})}) \leq (5\sqrt{n}+1)\cdot \weight(\algset)$.
    Therefore, $\frac{\opt(\instance)}{\alg(\instance)} = \frac{\weight(\optset)}{\weight(\algset)} \leq \frac{\exclusiveWeight(\optset)}{\weight(\algset)} + 1 = \frac{\sum_{P^{(\boldsymbol{\cdot})}\in \partition} \exclusiveWeight(P^{(\boldsymbol{\cdot})})}{\weight(\algset)} + 1 \leq \frac{(5\sqrt{n}+1)\cdot\weight(\algset)}{\weight(\algset)} + 1 = 5\sqrt{n}+2$.
\end{proof}

\section{Cactus graphs}
\label{Sec:cactus}
\subsection{$\algCactus$: Algorithm for cactus graphs}
Similar to the previous case, $\algCactus$ proceeds in rounds, each consisting of several iterations.
In each iteration, if there are enough firefighters to protect the heaviest component adjacent to the fire source, $\algCactus$ makes the greedy choice.
Moreover, if only one firefighter is available but the greedy choice guarantees a gain of at least $\sqrt{\weight(\rootC_1)}$, where $\rootC_1$ is the heaviest root cycle in the current graph, the algorithm still follows the greedy choice.

The critical case occurs when only one firefighter is available, and the greedy choice does not guarantee a sufficiently large immediate gain.

In this situation, as in the strategy for $1$-almost trees, $\algCactus$ breaks a root cycle $C$ by protecting a vertex $u$ on the cycle that maximizes $\tol(u, \rootC, \sqrt{\weight(\rootC_1)})$, thereby ensuring that at least $\sqrt{\weight(\rootC_1)}$ vertices lie as far from the fire as possible (in terms of distance from~$\newR$).
Unlike in $1$-almost trees, however, a cactus graph may contain multiple root cycles.
A difficulty arises because if the algorithm repeatedly breaks cycles due to the lack of immediate gain, it may fail to secure many vertices that could otherwise be saved.
To address this, we introduce a \emph{cool-down timer} mechanism: once a cycle is broken, the algorithm enters a cool-down period during which it must follow the greedy choice whenever new firefighters become available.
The rule for selecting a vertex to break a cycle is thus adapted to this mechanism.
The detailed description is in Algorithm~\ref{Alg:cactus}.

\runtitle{Cool-down mechanism and \cactusBreakCycle.}
We first find the cycle vertex adjacent to the fire source $\newR$ that has the largest tolerance in the sense that if the edge between it and $\newR$ is removed, the distance that the fire can spread while keeping $\sqrt{\weight(\rootC_1)}$ vertices remaining unburned is maximized.
Let $\Cmax$ be the cycle containing this vertex $\umax$ with the largest tolerance $\dmax$.
Now the \cactusBreakCycle algorithm tries to find an alternative cycle vertex on $\Cmax$ while keeping the invariant ``\emph{at least $\sqrt{\weight(\rootC_1)}-\weight(u)$ vertices in $\newT{\Cmax}{(r,\umax)}$ remain unburned for $\dmax$ rounds''.}
This alternative chooses a break-cycle vertex farther from the fire source, while exposing only a limited number of additional vertices to the risk of being burned compared to breaking at a neighbor of $\newR$.
We show in Lemma~\ref{lem:CooldownIsFine} that our cycle-breaking decision according to \cactusBreakCycle, which protects $\uCutCycle$ on cycle $\Cmax$, saves at least a $\tfrac{1}{2\sqrt{n}}$ fraction of the vertices that would be saved by breaking at any other cycle vertex on any cycle.
The setting of the cool-down timer ensures that whenever a firefighter is released before the timer expires, the greedy choice protects at least as many vertices as the remaining available ones on the path that defines the timer.

To identify a suitable alternative break-cycle vertex, we introduce the notations $\EnewT{\rootC}{u,v}$ and $\etol((u,v), \rootC, m)$.
Here, $\EnewT{\rootC}{u,v}$ denotes the subgraph induced by $\CoveredSet(\rootC)$ with the edge $(u,v)$ removed.
The notation $\etol((u,v), \rootC, m)$ then refers to the tolerance with respect to $m$ in $\EnewT{\rootC}{u,v}$.
The detailed definition of \cactusBreakCycle is given in Algorithm~\ref{Alg:ImprovedBreak}.

\begin{algorithm}[t]
    \caption{Procedure of $\algCactus$ for cactus graph in round $i$} 
    \label{Alg:cactus}
    \textbf{Input:} Graph $\newG$ with fire source $\newR$ and $f_i$ available firefighters\\
    $\cool \gets \max\{0, \cool-1\}$\\
    \While{$f_i > 0$ and there are available vertices in $\newG$}{
        \eIf{$\newR$ is not a cycle vertex in $\newG$}{
            Protect the vertex $v$ in $N(\newR)$ with the highest weight.\\
            Set $f_i\gets f_i -1$, $\newG \gets \newG \setminus \CoveredSet(v)$.
        }{
            \tcp{There are root cycles $\rootC_1, \rootC_2, \cdots$ containing $\newR$ and $\weight(\rootC_1) \geq \weight(\rootC_2) \geq \cdots$}
            Consider unprotected vertices $v_1, v_2, \cdots$ in $N(\newR)\cup (\bigcup_h\rootC_h)$ with $\weight(v_1) \geq \weight(v_2) \geq \cdots$

            \eIf{$f_i \geq 2$}{
                \eIf{$\weight(v_1) + \weight(v_2) \geq \weight(\rootC_1)$}{   
                    Protect $v_1$.\\
                    Set $f_i \gets f_i - 1$, $\newG \gets \newG \setminus \CoveredSet(v_1)$.
                }{
                    Protect both vertices $v_\ell, v_r$ in $N(\newR) \cap \rootC_1$. \label{line:cactusSimultaniously}\\
                    Set $f_i \gets f_i - 2$, $\newG \gets \newG \setminus \CoveredSet(\{v_\ell, v_r\})$.
                }
            }{
                \eIf{$\weight(v_1) \geq \sqrt{\weight(\rootC_1)}$ or $\weight(\rootC_1) \leq \sqrt{n}$ or $\emph{\cool} > 0$\label{Alg:cactus:line:greedy_with_one_firefighter}}{
                        Protect $v_1$.\\
                        $\cool \gets 0$.\\
                        Set $f_i \gets f_i - 1$, $\newG \gets \newG \setminus \CoveredSet(v_1)$.
                }{
                    $\uCutCycle, \cool \gets$ \cactusBreakCycle$(\newG, \newR, \sqrt{n})$ \\
                    Protect $\uCutCycle$.\\
                    Set $f_i \gets f_i - 1$, $\newG \gets \newG \setminus \CoveredSet(\uCutCycle)$.
                }
            }
        }
    }
\end{algorithm}

\begin{algorithm}[t]
    \caption{Procedure \cactusBreakCycle$(G, r, \eta)$} 
    \label{Alg:ImprovedBreak}
    \textbf{Output:} a cycle vertex $\uCutCycle$ and a cool-down distance $\cool$\\
    Let $\rootC_1, \rootC_2, \cdots$ be the root cycles with weight at least $\eta$, where $\weight(C_1) \geq \weight(C_2) \geq \cdots \geq \eta$\\
    Consider all cycle vertices $u$ in $N(r) \cap \rootC_h$ for some root cycle $\rootC_h$ such that $\weight(\rootC_h)-\weight(u) \geq \sqrt{\weight(\rootC_1)}$\\

    Let $\umax \in N(r)\cap\rootC_h$ be the cycle vertex with the largest $\etol((\umax, r), \rootC_h, \sqrt{\weight(\rootC_1)})$ (where $\umax$ is included in root cycle $\rootC_h$), $\dmax$ be the tolerance $\etol((\umax, r), \rootC_h, \sqrt{\weight(\rootC_1)})$, and $\Cmax$ be $\rootC_h$, where $\rootC_h = (r, u_1, u_2, \cdots, u_p)$ in a cyclic order and $u_1 = \umax$\\

    \medskip

    \For{$\uCutCycle = u_1, u_2, \cdots$}{
        \If{there is $v \in \CoveredSet(\uCutCycle)$ such that $\emph{\dist}(\EnewT{\Cmax}{\umax,r}, r, v) \geq \dmax$}{
            \Return{$\uCutCycle$, $\emph{\dist}(\EnewT{\Cmax}{\umax,r}, r, \uCutCycle)$}
        }
    }
\end{algorithm}

\begin{observation}[Properties of $\algCactus$]\label{obs:algCactus}
    When the input graph is a cactus graph, by the definition of $\algCactus$:
    \begin{enumerate}[(a)]
        \item \label{obs:algCactusNoChain}
        Consider all cycles intersecting a path from $r$ to an arbitrary vertex $v$. 
        The only cycle that is possible to be broken by \cactusBreakCycle is the one closest to $r$. 
        This is because, for every vertex $u$ on this cycle $C$, we have $\weight(u) \leq \sqrt{n}$, which is insufficient to invoke another \cactusBreakCycle procedure.
        
        \item \label{obs:algCactusBreakCycle}
        Assume that in round $i$, the algorithm $\algCactus$ breaks a cycle with a cool-down value $\cool$ chosen by \cactusBreakCycle. 
        Let $i^\prime > i$ be the next round when there are firefighters available.
        If $i^\prime - i \leq \cool$, in round $i^\prime$, $\algCactus$ greedily protects the heaviest vertex $v_1^\prime$ in the reduced graph $\newG^\prime$.

        \item Each cycle $C$ can have at most three vertices protected by $\algCactus$.
    \end{enumerate}
\end{observation}

\paragraph*{Quality of the choice of cycle-breaking vertex $\uCutCycle$}
We first show the feasibility of procedure \cactusBreakCycle. 
More specifically, we show that the procedure indeed returns an alternative cycle vertex $\uCutCycle$ on the cycle $\CCutCycle$, such that protecting $\uCutCycle$ ensures that there are at least $\sqrt{\weight(\rootC_1)}-\weight(\uCutCycle)$ vertices remaining a distance of $\dmax$ away from the fire source.

\begin{figure}[ht]
    \centering
    \begin{subfigure}[b]{0.45\textwidth}
        \centering
        \includegraphics[width=\linewidth]{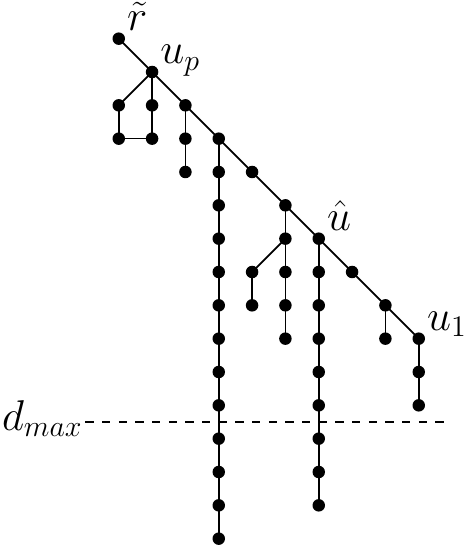}
        \caption{Vertex $\uCutCycle$ is chosen such that no vertices further ``down'' the cycle have a vertex in their covered set at distance at least $\dmax$ from the root $\newR$.}
        \label{fig:lemma16}
    \end{subfigure}%
    \hfill%
    \begin{subfigure}[b]{0.45\textwidth}
        \centering
        \includegraphics[width=\linewidth]{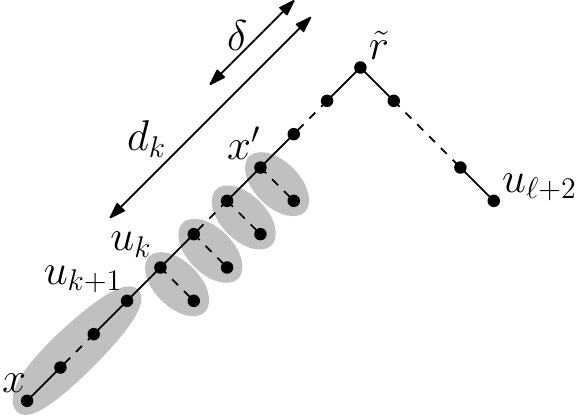}
        \caption{The weight $\weight(\{x, x^\prime \} )$ is made up by the weights of the $(d_k - \delta)$ covered sets of size at most $\sqrt{\weight(\rootC)}$, and the weight of an additional covered set $\CoveredSet(\{ u_{k+1}, x \})$.}
        \label{fig:lemma18}
    \end{subfigure}
    \hfill
    \caption{}
    \label{fig:lemmas1618}
\end{figure}

\begin{lemma}[Feasibility of $\cactusBreakCycle$]\label{lem:alternativeIsFine}
    Let $\dmax$ be the maximum tolerance described in Algorithm~\ref{Alg:ImprovedBreak}. 
    The procedure of \cactusBreakCycle{} always terminates and returns a cycle vertex $\uCutCycle$ with $\amountBelow(\newT{\CCutCycle}{\uCutCycle}, d_{\max}) \geq \sqrt{\weight(\rootC_1)} - \weight(\uCutCycle)$.
\end{lemma}

\begin{proof}
    \runtitle{Termination.}
    Based on the definition of \cactusBreakCycle, we focus on all cycle vertices $u$ in the neighborhood of $\newR$ that are ``not too heavy''.
    More specifically, for every root cycle $\rootC_h$, we consider the vertices $u \in N(\newR) \cap \rootC_h$ with $\weight(u) \leq \weight(\rootC_h)-\sqrt{\weight(\rootC_1)}$, where $\rootC_1$ is the heaviest root cycle. 
    Note that since we only reach this procedure when $\weight(v_1) < \sqrt{\weight(\rootC_1)}$, at least two cycle vertices satisfy this condition: the vertices $u \in N(\newR) \cap \rootC_1$.
    
    When \cactusBreakCycle{} has selected a vertex $\umax \in N(\newR) \cap \rootC_h$, there are at least $\sqrt{\weight(\rootC_1)}$ vertices at a distance $d_{\max}$ from $\newR$ in the graph $\EnewT{\rootC_h}{\umax,\newR}$.
    Also, every vertex $v \in \EnewT{\rootC_h}{\umax,\newR}$ is in the covered set of some $u_k \in \rootC_h \setminus \{ \newR \}$.
    Since there must exist a vertex $v$ at distance at least $d_{\max}$ such that $v \in \CoveredSet(\uCutCycle)$ for some $\uCutCycle \in \{ u_1, \dots u_p \}$, \cactusBreakCycle{} terminates.
   
    \runtitle{Selection of $\uCutCycle$.}
    By the definition of tolerance, since $\etol((u_1, \newR), \rootC_h, \sqrt{\weight(\rootC_1)}) = d_{\max}$, it follows that $\amountBelow(\EnewT{\rootC_h}{u_1,\newR}, d_{\max})\geq \sqrt{\weight(\rootC_1)}$.
    Let $\uCutCycle \in \rootC_h$ be the vertex returned by $\cactusBreakCycle$.
    Then, none of the vertices $u_k \in \{ u_1, u_2, \dots, \uCutCycle -1 \}$ were returned by $\cactusBreakCycle$.
    Therefore, there cannot exists a vertex $v_k \in \CoveredSet(u_k)$ such that $\dist(\EnewT{\rootC_h}{u_1,\newR}, \newR, v_k) \geq \dmax$.
    Thus, all at least $\sqrt{\weight(\rootC_1)}$ vertices $v_m$ that have distance at least $d_{\max}$ to the root $\newR$ in graph $\EnewT{\rootC_h}{u_1, \newR}$ must be in the covered sets of vertices $u_m \in \{ \uCutCycle, \uCutCycle + 1, \dots, u_p \}$, as illustrated in Figure~\ref{fig:lemma16}.
    Observe that the distance for such a vertex $v_m \in \CoveredSet(u_m)$, for $u_m \neq \uCutCycle$, to the root $\newR$ in the graph $\newT{\rootC_h}{\uCutCycle}$ is the same as in the graph $\EnewT{\rootC_h}{u_1,\newR}$.
    Then, combined with the fact that $\amountBelow(\EnewT{\rootC_h}{u_1, \newR}, d_{\max})\geq \sqrt{\weight(\rootC_1)}$, it follows that $\amountBelow(\newT{\rootC_h}{\uCutCycle}, d_{\max}) \geq \sqrt{\weight(\rootC_1)} - \weight(\uCutCycle)$.
\end{proof}

Following Lemma~\ref{lem:alternativeIsFine}, we show that the ``quality'' of protecting the alternative vertex $\uCutCycle$, measured by the number of vertices that remain at (shortest) distance $\dmax$ from the fire source, is within a bounded factor of the quality obtained by breaking any other cycle.

\begin{lemma}[Secured alternative cycle break]\label{lem:CooldownIsFine}
    Assume that the algorithm $\algCactus$ protects $\uCutCycle$ to break the cycle $\CCutCycle$, where $\CCutCycle$ contains $\uCutCycle$.
    Then, for any positive integer $d$, any root cycle $\rootC$
    , and every vertex $u \in \rootC$, $\frac{\amountBelow(\newT{\rootC}{u}, d)}{\amountBelow(\newT{\CCutCycle}{\uCutCycle}, d) + \weight(\uCutCycle)} \leq 2\sqrt{n}$, where $n$ is the number of vertices in the original input graph $G$.
\end{lemma}

\begin{proof}
    If $\weight(\rootC) \leq \sqrt{n}$, then the bound immediately follows. Hence, assume $\weight(\rootC) > \sqrt{n}$.

    Note that $\CCutCycle = \Cmax$, the cycle selected by the \cactusBreakCycle algorithm. 
    We denote the vertices in $\CCutCycle$ in a cyclic order $(\newR, u_1, u_2, \cdots, u_p)$, where $u_1 = \umax$, the vertex in $N(\newR) \cap \CCutCycle$ with the largest tolerance.
    This tolerance is then denoted by $\dmax = \etol((u_1, \newR), \CCutCycle, \sqrt{\weight(\rootC_1)})$.
    Therefore, $\amountBelow(\EnewT{\CCutCycle}{u_1, \newR}, \dmax) \geq \sqrt{\weight(\rootC_1)}$.
    For arbitrary $d$ and a cycle vertex $u \in \rootC$, we consider two cases: 1) $d > \dmax$ and 2) $d \leq \dmax$.

    \runtitle{Case $\mathbf{d > \dmax}$.}
    Recall that, by definition, the maximum tolerance $\dmax$ is the maximum depth such that there exists a root cycle $\rootC^\prime$ and a vertex $u^\prime \in N(\newR) \cap \rootC^\prime$, such that $\amountBelow(\EnewT{\rootC^\prime}{u^\prime, \newR}, \dmax) \geq \sqrt{\weight(\rootC_1)}$.
    Thus, since $d > \dmax$, for any cycle $\rootC$ and any $u_x \in N(\newR) \cap \rootC$, $\amountBelow(\newT{C}{u_x}, d) \leq \amountBelow(\EnewT{C}{u_x, \newR}, d) <\sqrt{\weight(\rootC_1})$.
    Then, it follows from Lemma~\ref{lem:NeighborsAreTheBest}, that for any $u\in \rootC$, 
    $\amountBelow(\newT{\rootC}{u}, d)$ is upper bounded by the maximum of $\amountBelow(\newT{\rootC}{u_1},d) + \weight(u_1)$ and $\amountBelow(\newT{\rootC}{u_p},d) + \weight(u_p)$.
    Let $u_x$ denote the vertex where this expression is maximized.
    Then, $\amountBelow(\newT{\rootC}{u}, d) \leq \amountBelow(\newT{\rootC}{u_x}, d) + \weight(u_x) \leq 2\sqrt{\weight(\rootC_1)} \leq 2\sqrt{n}$, where the second-last inequality is due to the condition of \cactusBreakCycle where $\weight(u_x) \leq \sqrt{\weight(\rootC_1)}$.

    \medskip

    \runtitle{Case $\mathbf{d \leq \dmax}$.}
    Similar to the proof of Lemma~\ref{lem:CycleBreakIsGood}, we first show that in this case, every $v \notin \CoveredSet(\uCutCycle)$ that contributes to $\amountBelow(\EnewT{\CCutCycle}{\umax, \newR}, \dmax)$ also contributes to $\amountBelow(\newT{\CCutCycle}{\uCutCycle}, \dmax)$.
    By the fact that such a vertex $v$ contributes to $\amountBelow(\EnewT{\CCutCycle}{\umax, \newR}, \dmax)$, its shortest distance to $\newR$ in $\EnewT{\CCutCycle}{\umax, \newR}$ is at least $\dmax$. 
    Since the edge $(\umax, \newR)$ is not in the induced subgraph of $\EnewT{\CCutCycle}{\umax, \newR}$, the shortest path from $v$ to $\newR$ in $\EnewT{\CCutCycle}{\umax, \newR}$ must not go through $\umax$. 
    Consider the shortest path from $v$ to $\newR$ in the subgraph induced by $\newT{\CCutCycle}{\uCutCycle}$. 
    Since $v \notin \CoveredSet(\uCutCycle)$, $v$ is still connected to $\newR$, and there is a (shortest) path from $v$ to $\newR$ in $\newT{\CCutCycle}{\uCutCycle}$. 
    Note that this shortest path does not go through $\umax$, since otherwise the first cycle vertex $v_h$ on this path going from $\umax$ to $\newR$ will be chosen as $\uCutCycle$ (see Figure~\ref{fig:lemma16}), which contradicts to the fact that $v \notin \CoveredSet(\uCutCycle)$.
    Given that the shortest path does not go through $\umax$, the shortest path is the same as the shortest path from $v$ to $\newR$ in $\EnewT{\CCutCycle}{\umax, \newR}$, and has the same distance of at least $\dmax$.
    Thus, $v$ also contributes to $\amountBelow(\newT{\CCutCycle}{\uCutCycle}, \dmax)$.
    Together with the fact that all vertices in $\CoveredSet(\uCutCycle)$ contribute to $\amountBelow(\EnewT{\CCutCycle}{\umax, \newR}, \dmax)$ by the selection of $\uCutCycle$, we conclude that $\amountBelow(\newT{\CCutCycle}{\uCutCycle}, \dmax) + \weight(\uCutCycle) \geq \amountBelow(\EnewT{\CCutCycle}{\umax, \newR}, \dmax) \geq \sqrt{\weight(\rootC_1)}$.
    By the monotonicity, $\amountBelow(\newT{\CCutCycle}{\uCutCycle}, d) + \weight(\uCutCycle) \geq \amountBelow(\newT{\CCutCycle}{\uCutCycle}, \dmax) + \weight(\uCutCycle) \geq \sqrt{\weight(\rootC_1)}$.
    The lemma is then proven since $\amountBelow(\newT{C}{u}, d) \leq \weight(\rootC) \leq \weight(\rootC_1)$.
\end{proof}

Next, we show that during the cool-down period set by \cactusBreakCycle, the number of vertices saved by the forced greedy choice, together with the protection of the alternative cycle-breaking vertex $\uCutCycle$, is at least a $\tfrac{1}{\sqrt{n}}$ fraction of the number of vertices saved by protecting any possible pair of vertices.

\begin{lemma}[Protection quality in cool-down period]\label{lem:CactusBreakCycleIsGood}
    Let $\uCutCycle$ and $\cool$ be the cycle vertex and cool-down distance returned by \cactusBreakCycle{} in round $i$. 
    Let $i^\prime > i$ be the next round with available firefighters. 
    If $i^\prime -i \leq \cool$, let $u^\prime$ be the first vertex protected by $\algCactus$ in round $i^\prime$. 
    For arbitrary vertices $x$ and $x^\prime$ that are available in round $i$ and round $i^\prime$, respectively, we have $\frac{\weight(\{x, x^\prime\})}{\weight(\{\uCutCycle, u^\prime\})} \leq 2 \sqrt{n}$, where $n$ is the number of vertices in the original graph $G$.
\end{lemma}

\begin{proof}
    Let $\newG$ be the graph at the moment that \cactusBreakCycle is called, and $\hat{G} = \newG\setminus\CoveredSet(\uCutCycle)$.
    Define $\umax$ and $\dmax$ according to the description in $\cactusBreakCycle$.
    Note that the process \cactusBreakCycle is only called when for every vertex $v$, $\weight(v) < \sqrt{\weight(\rootC_1)}$ and $\weight(\rootC_1) > \sqrt{n}$.
    Let $\delta = i^\prime - i$, which is the maximum distance the fire can spread between turns $i$ and $i^\prime$.
    In the following, we show that $\weight(\{\uCutCycle, u^\prime\}) \geq \dmax-\delta$ and $\weight(\{x, x^\prime\}) \leq (\dmax-\delta)\cdot \sqrt{\weight(\rootC_1)}$. 
    Then, the lemma is proven.

    \medskip
    
    \runtitle{A lower bound for $\weight(\{\uCutCycle, u^\prime\})$.}
    By the choice of \cactusBreakCycle, there exists a vertex $v\in \CoveredSet(\uCutCycle)$ at distance $\dmax$ from $\newR$ in $\EnewT{\newG}{\umax,\newR}$ (here we abuse the notation so that $\EnewT{\newG}{\umax,\newR}$ is the graph $\newG$ with edge $(\umax, \newR)$ removed). 
    That is, $\dist(\EnewT{\newG}{\umax, \newR}, \newR, v) \geq \dmax$.
    Since $v \in \CoveredSet(\uCutCycle)$, any path from $v$ to $\newR$ must contain $\uCutCycle$. 
    Thus, $\dist(\EnewT{\newG}{\umax, \newR}, \newR, v) = \dist(\EnewT{\newG}{\umax, \newR}, \newR, \uCutCycle) + \dist(\EnewT{\newG}{\umax,\newR}, \uCutCycle, v)$.

    Recall that $\delta = i^\prime - i$ is the maximum distance the fire can spread between turns $i$ and $i^\prime$.
    When $\uCutCycle$ is protected, at least $\dist(\EnewT{G}{\umax, \newR}, \uCutCycle, v)$ vertices were covered.
    Hence, if $\delta \geq \dist(\EnewT{G}{\umax,\newR}, \newR, \uCutCycle)$, at least $\dmax-\delta$ vertices are saved by protecting $\uCutCycle$ alone.
    
    In contrast, if $\delta < \dist(\EnewT{G}{\umax,\newR}, \newR, \uCutCycle)$, then there is at least one cycle vertex $u$ at distance $\delta$ from $\newR$ in $G\setminus \CoveredSet(\uCutCycle)$ is available to protect for $\algCactus$ at round $i^\prime$.
    That is, $\dist(G\setminus \CoveredSet(\uCutCycle), \newR, u) = \delta$.
    Note that in round $i^\prime$, $u$ has weight at least $\dist(\EnewT{G}{\umax, \newR}, u, \uCutCycle)$. 
    Moreover, by Observation~\ref{obs:algCactus}(\ref{obs:algCactusBreakCycle}), $u^\prime$ is the heaviest available vertex. 
    Thus, by protecting $u^\prime$, at least $\dist(\EnewT{G}{\umax,\newR}, u, \uCutCycle)$ vertices are covered.
    Hence, $\dmax \leq \dist(\EnewT{G}{\umax, \newR}, \newR, v) = \dist(\EnewT{G}{\umax, \newR}, \newR, u) + \dist(\EnewT{G}{\umax, \newR}, u, \uCutCycle) + \dist(\EnewT{G}{\umax, \newR}, \uCutCycle, v)$.
    Since, by the setting, $\dist(\EnewT{G}{\umax, \newR}, \newR, u) = \dist(G\setminus\CoveredSet(\uCutCycle), \newR, u)= \delta$, at least $\dist(\EnewT{G}{\umax, \newR}, u, \uCutCycle)$ vertices were covered in round $i^\prime$, and at least $\dist(\EnewT{G}{\umax, \newR}, \uCutCycle, v)$ vertices are covered in round $i$.
    Therefore, $\weight(\{\uCutCycle, u^\prime\}) \geq \dmax-\delta$.

    \medskip
    
    \runtitle{An upper bound for $\weight(\{x, x^\prime\})$.}
    Note that when \textsc{InprovedBreak} is called, for every vertex $v$, $\weight(v) < \sqrt{\weight(\rootC_1)}$.
    Thus, if $x$ and $x^\prime$ are not in the same cycle, then $\weight(\{x, x^\prime\}) = \weight(x) + \weight(x^\prime) \leq 2\sqrt{\weight(\rootC_1)}$, and the lemma follows. 
    Therefore, from now on, we consider the case where $x$ and $x^\prime$ are contained in the same cycle.
    First note that if $x$ and $x^\prime$ would be contained in a cycle that is not a root cycle, it would have been possible to save more than $\weight(\{x, x^\prime \})$ vertices with just a single protected vertex $\overline{x}$.
    Then, $\weight(\{x, x^\prime \}) \leq \weight(\overline{x}) < \sqrt{\weight(\rootC_1)}$.
    Therefore, $x$ and $x^\prime$ are contained in the same root cycle~$\rootC$.
    Second, note that, by the definition of $\dmax$, for any cycle $\rootC$ and two neighboring cycle vertices $u$ and $\mu$ on $\rootC$, 
    \begin{gather}\label{eq:mu}
        \text{if } \amountBelow(\EnewT{\rootC}{u, \mu}, d) \geq \sqrt{\weight(\rootC_1)} \text{, then } d \leq \dmax.
    \end{gather}
    Since the fire spreads for distance of $i^\prime-i = \delta$ between the two rounds, the distance between $x^\prime$ and $\newR$ is at least $\delta$ in $\newG\setminus \CoveredSet(x)$.

    Denote the vertices on the arc between $x^\prime$ and $x$ without going through $\newR$ as $(x^\prime, u_1, \cdots, u_\ell, x)$, $u_{\ell+2}$ as $x$, and the other neighbor of $x$ as $u_{\ell+2}$, which can be $\newR$.
    Consider every vertex $u_i$ and $\weight(\{x, u_i\})$. 
    If $\weight(\{x, u_i\}) < \sqrt{\weight(\rootC_1)}$ for all $u_i$, then the lemma holds.
    Otherwise, if there are some $u_i$ with $\weight(\{x, u_i\}) \geq \sqrt{\weight(\rootC_1)}$, consider the one that is the closest to $x$ and assume that it is $u_k$.
    Let $d_k$ be the distance from $u_k$ to $\newR$ in $\EnewT{\newG}{x, u_{\ell+2}}$. 
    That is, $d_k = \dist(\EnewT{\newG}{x, u_{\ell+2}}, \newR, u_k)$.
    Then, $\amountBelow(\EnewT{\newG}{x, u_{\ell+2}}, d_k) \geq \weight(x, u_k) \geq \sqrt{\weight(\rootC_1)}$.
    Thus, by equation \ref{eq:mu}, $d_k \leq \dmax$.
    Furthermore, observe that $ k = \dist(\EnewT{\newG}{x, u_{\ell+2}}, x^\prime, u_k) = \dist(\EnewT{\newG}{x, u_{\ell+2}}, \newR, u_k) -  \dist(\EnewT{\newG}{x, u_{\ell+2}}, \newR, x^\prime) = d_k - \delta$.
    
    Note that, by the selection of $u_k$, $\weight(\{u_{k+1}, x\}) < \sqrt{\weight(\rootC_1)}$.
    Then, as is also sketched in Figure~\ref{fig:lemma18}, $\weight(\{x, x^\prime \}) \leq  \left( \sum_{i=0}^k \weight(u_i) \right) + \weight(\{u_{k+1}, x\}) < \left( \sum_{i=0}^k \sqrt{\weight(\rootC_1)} \right) + \sqrt{\weight(\rootC_1)} = (d_k - \delta + 1) \cdot \sqrt{\weight(\rootC_1)} \leq (\dmax - \delta + 1) \cdot \sqrt{\weight(\rootC_1)}$.

    To conclude, by applying the lower bound for $\weight(\{\uCutCycle, u^\prime\})$ and the upper bound for $\weight(\{x, x^\prime \})$, we get $\frac{\weight(\{x, x^\prime\})}{\weight(\{\uCutCycle, u^\prime\})} \leq \frac{(\dmax - \delta + 1) \cdot \sqrt{\weight(\rootC_1)}}{\dmax - \delta} \leq 2\sqrt{\weight(\rootC_1)} \leq 2\sqrt{n}$.
\end{proof}

\subsection{Competitive analysis of $\algCactus$}
For analyzing $\algCactus$ on cactus graphs, we first introduce a concept of \emph{delay}.
A vertex $v$ is \emph{delayed} by another vertex $p$ if when $p$ is protected, $v$ is not covered, but the time that fire reaches $v$ is strictly delayed. 
In other words, a delayer $p$ of $v$ is on every shortest path from $r$ to $v$, but not on every path from $r$ to $v$.
That is, $p$ is $v$'s \emph{delayer} if $p$ is on every shortest path from $r$ to $v$ in $G$ and $v \notin \CoveredSet(p)$.
We generalize this concept and say that a vertex $v$ is delayed by a set of vertices $S$ if $S$ contains at least one delayer of $v$ and $S$ does not cover $v$ (that is, $v\notin \CoveredSet(S))$.
Moreover, given a vertex $p$, we define its \emph{delayed set} $\delayedSet(p)$ by the set of all vertices that are delayed by $p$.
For a set $S$ of vertices, its delayed set $\delayedSet(S)$ is defined by $\bigcup_{p\in S} \delayedSet(p)$.
Given $\vopt \in \optset$, we denote $\delayedSet(\vopt) \cap \optset$ as $\optDelayedSet(\vopt)$. 
That is, $\optDelayedSet(\vopt)$ is the set of vertices protected by $\opt$ and delayed by $\vopt$.

Recall that in algorithm $\algCactus$, at most three vertices can be protected on the same cycle.
We call such cycle vertices \emph{peers}.
In what follows, we may refer to a cycle vertex and its peers.

\paragraph*{Partition $\partition$ of $\optset$}
Given the set $\optset$ of vertices protected by $\opt$, which by Lemma~\ref{lem:non-redundant} is non-redundant.
We partition $\optset$ into the three \emph{types} of parts, where there can be multiple parts of Type $b$ or $c$.

\begin{enumerate}[Type a:]
    \item \runtitle{Non-cycle vertices that are not delayed.}
    This type has at most one part $\parta$, which consists of all non-cycle vertices in $\optset$ that are not in $\delayedSet(\optset)$.
    
    \item \runtitle{One cycle vertex with its delayed set.}
    For each cycle vertex $\vopt_t \in \optset$ that is the only vertex in the cycle protected by $\optset$, $\vopt_t$ and the set of vertices protected by $\optset$ and delayed by $\vopt_t$ form a part.
    That is, given such a vertex $\vopt_t$, they form a Type-$b$ part $\partb_t = \{\vopt_t\} \cup \optDelayedSet(\vopt_t)$.
    Note that $\partb_t$ might be empty for some $t$.
    We call $t$ the \emph{representative vertex} of $\partb_t$.
    
    \item \runtitle{Two cycle vertices with their delayed set.}
    For each cycle $C$ that contains two vertices $\vopt_t$ and $\vopt_{s}$, and at least one of $\vopt_t$ and $\vopt_{s}$ is not delayed by $\optset$, these two vertices and their delayed sets form a Type-$c$ part $\partc_{t,s}$. 
    That is, $\partc_{t,s} = \{\vopt_t, \vopt_{s}\} \cup \optDelayedSet(\vopt_t) \cup \optDelayedSet(\vopt_{s})$.
    For some pair of $t$ and $s$, $\partc_{t, s}$ might be empty.
    We call $\vopt_t$ and $\vopt_s$ the \emph{representative vertices} of $\partc_{t,s}$.
\end{enumerate}

A vertex might be eligible to belong to both a Type-$b$ part and a Type-$c$ part.
In such cases, we give priority to assign it to a Type-$c$ part.
If a vertex $\vopt_t \in \optset$ belongs to at least two cycles, observe that the fire can approach $\vopt_t$ from only one of these cycles. 
Consequently, protecting $\vopt_t$ suffices to save all other cycles entirely, and no additional vertex needs to be protected in those cycles.
The only cycle in which two vertices might be protected is the cycle from which the fire approached $\vopt_t$, hence this cycle determines whether $\vopt_t$ is assigned to a Type-$b$ or a Type-$c$ part.

\begin{lemma}
    The partition $\mathcal{P}^*$ is a proper, cycle-respecting partition of $\optset$.
\end{lemma}

\begin{proof}
    By Lemma~\ref{lem:non-redundant}, a non-redundant optimal solution on a cactus graph protects at most two vertices on the same cycle.
    According to the definition of Type-$c$ parts, $\partition$ is cycle-respecting.
    
    By the definition of delayers that cannot cover the delayed vertex, in a cactus graph, a vertex cannot be delayed by a non-cycle vertex. 
    Moreover, every delayed vertex is delayed by exactly one non-delayed vertex $\vopt_p \in \optset$ (since we define the delayers as strictly delaying the time the fire will reach).
    Then, either there is no other vertex protected by $\opt$ in the cycle that contains $\vopt_p$, or the other vertex protected by $\opt$ in this cycle was also delayed by $\vopt_p$.
    Thus, any delayed non-cycle vertex is ``attached'' to its delayer, which, in the former case, is a single cycle vertex on a cycle, and belongs to a Type-$b$ part, and in the latter case, is a pair of protected cycle vertices of which at least one is not delayed, and belongs to a Type-$c$ part.
    Furthermore, consider any pair of cycle vertices on the same cycle that are both protected by $\opt$ and delayed by another cycle vertex in $\optset$.
    These two vertices will also belong to the Type-$b$ part that corresponds to their delayer.
    Therefore, every vertex in $\optset$ is captured by exactly one part.
\end{proof}

\paragraph*{Charging function $\charging$}
We define the charging function $\charging: \partition \to 2^{\algset}$ according to the type of parts.

\runtitle{Type-$\mathbf{a}$ part.}
If the Type-$a$ part $\parta$ exists, it is charged to $\charging(\parta)$, which consists of the vertices $\valg_t$ corresponding to all $\vopt_t \in \parta$.
Moreover, if $\valg_t$ is a cycle vertex, then all of its peers on the same cycle are also included in $\charging(\parta)$.
That is, $\charging(\parta) = \bigcup_{\vopt_t \in \parta}\bigl(\{\valg_t\} \cup \{\text{peers of } \valg_t \mid \valg_t \text{ is a cycle vertex}\}\bigl)$.

\runtitle{Type-$\mathbf{b}$ and Type-$\mathbf{c}$ parts.}
Similar to the Type-$a$ part, for every Type-$\mathbf{b}$ or Type-$\mathbf{c}$ $P$, its $\charging(P)$ consists of all $\valg_t$ corresponding to every $\vopt_t \in P$, together with their peers (if any).
Additionally, for Type-$b$ and Type-$c$ parts $P$, for every $\vopt_t \in P$, we further include two additional sets of vertices:
\begin{itemize}
    \item \texttt{Break-cycle rule}: If $\valg_t$ is protected via \cactusBreakCycle, we also include all vertices covered by $C$, where $C$ is the cycle containing $\valg_t$.
    
    \item \texttt{Cool-down rule}: If $\valg_t$ protects the heaviest vertex at time $t$ due to a positive $\cool$ timer, which was set in the previous break-cycle at time $s < t$, we also include $\valg_s$ and all vertices covered by $C_s$ containing $\valg_s$. 
\end{itemize}

By a rough approximation based on how the charged subset of $\algset$ is defined, each vertex in $\algset$ can only be charged by at most $13$ parts.
By a more careful counting, we get:

\begin{observation}\label{obs:cactusBeta}
    Any vertex in $\algset$ is charged at most $5$ times.
\end{observation}

\begin{proof}
    We analyze the number of parts a vertex $\valg_t \in \algset$ can belong to.
    By Observation~\ref{obs:algCactus}(\ref{obs:algCactusNoChain}) and the fact that a cycle can only be broken once, any vertex $\valg \in \algset$ can be included in at most one Type-$b$ or Type-$c$ part using the \texttt{Break-cycle rule}.
    Moreover, due to the linear time labeling, at most one Type-$b$ or Type-$c$ part can include $\valg$ via the \texttt{Cool-down rule}.    
    
    If $\valg_t$ is a non-cycle vertex, it can be charged by $\vopt_t$. 
    Hence, together with the parts it can belong to via the \texttt{Break-cycle rule} and the \texttt{Cool-down rule}, $\valg_t$ can belong to at most $3$ parts.
    In contrast, any cycle vertex $\valg_t$ has at most two peers, each of which may contribute a part containing $\valg_t$.
    Together with the parts it is included by \texttt{Break-cycle rule} and \texttt{Cool-down rule}, a cycle vertex $\valg_t$ can belong to at most $5$ parts.
\end{proof}

Next, we analyze the exclusive weight of a part by its type.

\begin{lemma}[Non-cycle vertices in $\optset$ that are not delayed]\label{lem:algCactusNonCycleVertices}
    If $\parta$ is the Type-$a$ part, then $\exclusiveWeight(\parta) \leq \sqrt{n}\cdot \weight(\charging(\parta))$.
\end{lemma}

\begin{proof}
    Recall that $\parta$ has all non-cycle vertices $\vopt_t \in \optset$ that are not delayed by $\optset$.
    Since $\vopt_t$ is not delayed by $\optset$, it is available as an option to $\algCactus$ at time $t$. 
    Consider the graph $\newG_t$ that $\algCactus$ sees at time $t$. 
    Let $\rootC_1$ and $v_1$ be the heaviest root cycle and the heaviest vertex in $\newG_t$. 
    By Lemma~\ref{lem:exclusiveWeight}, $\exclusiveWeight(\vopt_t) \leq \weight_t(\vopt_t) \leq \weight_t(v_1)$.
    The vertex $\valg_t$ protected by $\algCactus$ at time $t$ can be chosen as 1) the heaviest vertex in $\newG_t$, 2) one of the two vertices for protecting $\rootC_1$, or 3) the vertex $\uCutCycle$ returned by \cactusBreakCycle, which breaks cycle $\CCutCycle$.

    In the first case, since $\valg_t$ is $v_1$, $\exclusiveWeight(\vopt_t) \leq \weight_t(\valg_t)$. 
    In the second case, let $\valg_{t+1}$ be the other vertex that protects $\rootC_1$ together with $\valg_t$.
    Observe that the second case only happens if $\weight_t(\rootC_1) > \weight_t(v_1) + \weight_t(v_2)$, where $v_2$ is the second heaviest vertex in $\newG_t$. 
    In the case where $\vopt_{t+1}$ is not in $\parta$, we have $\exclusiveWeight(\vopt_t) \leq \weight_t(v_1) < \weight_t(\rootC_1) = \weight_t(\{\valg_t, \valg_{t+1}\})$.
    Similarly, in the case where $\vopt_{t+1} \in \parta$, as, by the definition of the partition, neither $\vopt_t$ nor $\vopt_{t+1}$ are cycle vertices, $\exclusiveWeight(\{ \vopt_t, \vopt_{t+1} \}) \leq \weight_t(\{ \vopt_t, \vopt_{t+1} \}) = \weight_t(\vopt_t) + \weight_t(\vopt_{t+1}) < \weight_t(\rootC_1) = \weight_t(\{\valg_t, \valg_{t+1}\})$.
    By the algorithm definition, the third case only happens when $\weight_t(v_1)^2 \leq \weight_t(\rootC_1)$.
    Thus, $\exclusiveWeight(\vopt_t) \leq \weight_t(\vopt_t) \leq \weight_t(v_1) \leq \sqrt{\weight_t(\rootC_1)} \leq \sqrt{n}$.
    Therefore, $\exclusiveWeight(\vopt_t) \leq \sqrt{n} \cdot \weight_t(\valg_t)$.
    
    Recall the charging function; $\charging(\parta) = \bigcup_{\vopt_t \in \parta}\bigl(\{\valg_t\} \cup \{\text{peers of } \valg_t \}\bigl)$.
    In order not to count the weight of vertices in $\parta$ multiple times, we define sets $S_t^* \subseteq \parta$ and $S_t^\texttt{A} \subseteq \charging(\parta)$.
    In the derivations above we have shown for subsets (which can be singleton sets) of $\parta$, which we can charge them to subsets of $\charging(\parta)$.
    However, this subset of $\parta$ might contain multiple vertices $\vopt_t, \vopt_s, \dots$ in $\parta$.
    Given $t< s$, we want to count the weight of this subset (as well as the weight of its image) just once, when we consider the exclusive weight of the vertex protected at time step $t$, and hence not count the weight of this entire subset again whenever we count the exclusive weight of the vertex protected at time step $s$.
    Thus, we define $S_t^* \subseteq \parta$ as follows;
    \[S_t^* = \begin{cases}
    \emptyset & \exists \valg_s : \valg_s \text{ is a peer of } \valg_t \text{ and } s > t.\\
    \{\vopt_t \} \cup \{ \vopt_s \in \parta \mid \valg_s \text{ is a peer of } \valg_t \} & \text{otherwise.}
    \end{cases}\]
    Similarly, we define $S_t^\texttt{A} \subseteq \charging(\parta)$ as;
    \[S_t^\texttt{A} = \begin{cases}
    \emptyset & \exists \valg_s : \valg_s \text{ is a peer of } \valg_t \text{ and } s > t.\\
    \{\valg_t \} \cup \{ \valg_s \mid \valg_s \text{ is a peer of } \valg_t \} & \text{otherwise.}
    \end{cases}\]
    Then, we have $\exclusiveWeight(\parta) = \sum_{t:\vopt_t \in \parta} \exclusiveWeight(S_t) \leq \sum_{t:\vopt_t \in \parta} \sqrt{n} \cdot \weight_t(S_t^\texttt{A}) \leq \sqrt{n} \cdot \weight(\charging(\parta))$.
\end{proof}

For analyzing Type-$b$ or Type-$c$ parts, we first show that the critical case happens when \cactusBreakCycle breaks a cycle in $\newG_t$ and \opt also protects the first vertex it protects on a cycle at time $t$.

\begin{lemma}
    \label{lem:typeBCycleCut}
    If $\partb_{t}$ is a Type-$b$ part, and at time $t$ $\algCactus$ calls \cactusBreakCycle, then $\exclusiveWeight(\partb_{t}) \leq 2\sqrt{n}\cdot \weight(\charging(\partb_{t}))$.
\end{lemma}

\begin{proof}
    Given a Type-$b$ part $\partb_t = \{\vopt_t\} \cup \optDelayedSet(\vopt_t)$, where $\optDelayedSet(\vopt_t)$ is the set of vertices in $\optset$ that are delayed by $\vopt_t$ (that is, $\optDelayedSet(\vopt_t) = \delayedSet(\vopt_t) \cup \optset$).
    Let $C$ be the cycle containing $\vopt_t$. 
    We consider the graph $\newG_t$ that $\algCactus$ sees at time $t$, where $v_1$ and $\rootC_1$ are the heaviest vertex and the heaviest cycle in $\newG_t$, respectively.

    We consider the case where $\valg_t$ is protected to break a cycle. 
    We analyze this case by considering the vertices protected by \opt after time $t$ up until $\algCactus$ finishes the game.
    The procedure \cactusBreakCycle only happens when $\weight_t(v_1) \leq \sqrt{\weight_t(\rootC_1)}$.
    Thus, for any vertex $v$ in $\newG_t$, $\weight_t(v) \leq \weight_t(v_1) \leq \sqrt{\weight_t(\rootC_1)}$.
    Since, by definition of $\partb_t$, $\vopt_t$ is the only vertex in $C$ that is protected by $\opt$, no vertex in $C$ is in $\optDelayedSet(\vopt_t)$. 
    Thus, every vertex in $\optDelayedSet(\vopt_t)$ has weight at most $\sqrt{\weight_t(\rootC_1)}$.
    Hence, if $\weight(\charging(\partb_t)) \geq |\optDelayedSet(\vopt_t)| + 1$, then $\frac{\exclusiveWeight(\partb_t)}{\weight(\charging(\partb_t))} \leq \sqrt{\weight(\rootC_1)} \leq \sqrt{n}$.

    In contrast, consider the scenario where $\weight(\charging(\partb_t)) < |\optDelayedSet(\vopt_t)| + 1$, which happens when $\algCactus$ finishes the game earlier than $\opt$ does.
    Recall that by the \texttt{Break-cycle rule}, $\charging(\partb_t)$ consists of all vertices in $\CoveredSet(C_x) \cap \algset$, where $C_x$ is the cycle broken by $\valg_t$.
    Thus, letting $i_\ell$ be the round $\algCactus$ finishes the game and $i_t$ be the round $\valg_t$ is protected, $\charging(\partb_t)$ covers at least $\amountBelow(\newT{C_x}{\valg_t}, i_\ell-i_t) + \weight(\valg_t)$ vertices.
    Meanwhile, after round $i_\ell$, $\optDelayedSet(\vopt_t)$ covers at most $\amountBelow(\newT{C}{\vopt_t}, i_\ell-i_t)$ since the fire spreads for distance of $i_\ell-i_t$.
    By Lemma~\ref{lem:CooldownIsFine}, letting $\delta$ be the number of vertices protected by \opt between rounds $i_t$ and $i_\ell$, $\frac{\exclusiveWeight(\partb_t)}{\weight(\charging(\partb_t))} \leq \frac{\delta\cdot\sqrt{\weight_t(\rootC_1)} + \amountBelow(\newT{C}{\vopt_t}, i_\ell - i_t)}{\delta + \amountBelow(\newT{C_x}{\valg_t}, i_\ell - i_t) + \weight_t(\valg_t)} \leq 2\sqrt{n}$.
\end{proof}

\begin{lemma}
    \label{lem:typeCCycleCut}
    If $\partc_{t, s}$ is a Type-$c$ part, and at time $t$ $\algCactus$ calls \cactusBreakCycle, then $\exclusiveWeight(\partc_{t, s}) \leq 4\sqrt{n}\cdot \weight(\charging(\partc_{t, s}))$.
\end{lemma}

\begin{proof}
    Given a Type-$c$ part $\partc_{t,s} = \{\vopt_t, \vopt_s\} \cup \optDelayedSet(\vopt_t) \cup \optDelayedSet(\vopt_s)$. 
    Without loss of generality, we assume that $t < s$.
    Let $C$ be the cycle containing $\vopt_t$ and $\vopt_s$.
    We consider the graph $\newG_t$ that $\algCactus$ sees at time $t$, where $v_1$ and $\rootC_1$ are the heaviest vertex and the heaviest cycle in $\newG_t$, respectively.

    Note that when \cactusBreakCycle is called, $\weight_t(v) \leq \sqrt{\weight_t(\rootC_1)}$ for all $v \in V$.
    Then, it follows from Lemma~\ref{lem:exclusiveWeight}, $\exclusiveWeight(v) \leq \sqrt{\weight_t(\rootC_1)}$ for every $v$.
    We split the analysis into two claims: 1) $\exclusiveWeight(\{\vopt_t, \vopt_s\}) \leq 2\sqrt{n}\cdot\weight(\charging(\partc_{t,s}))$ and 2) $\exclusiveWeight(\optDelayedSet(\vopt_t) \cup \optDelayedSet(\vopt_s)) \leq 2\sqrt{n}\cdot\weight(\charging(\partc_{t,s}))$.
    Since the partition is cycle-respecting, we can conclude that $\frac{\exclusiveWeight(\partc_{t, s})}{\weight(\charging((\partc_{t, s}))} \leq \frac{\exclusiveWeight(\{\vopt_t, \vopt_s\})+\exclusiveWeight(\optDelayedSet(\vopt_t) \cup \optDelayedSet(\vopt_s))}{\weight(\charging(\partc_{t, s}))} \leq 4\sqrt{n}$, and the lemma is proven. 
    In the following, we prove the two claims.

    \medskip

    \runtitle{Claim 1: $\exclusiveWeight(\{\vopt_t, \vopt_s\}) \leq \sqrt{n}\cdot\weight(\charging(\partc_{t,s}))$.}
    By plugging $\newG_t$ as $G$, $(\valg_t, \valg_s)$ as $(\uCutCycle, u^\prime)$, and $(\vopt_t, \vopt_s)$ as $(x, x^\prime)$ in Lemma~\ref{lem:CactusBreakCycleIsGood}, we get $\frac{\exclusiveWeight(\{\vopt_t, \vopt_s\})}{\weight(\charging(\partc_{t,s}))} \leq \frac{\weight_t(\{\vopt_t, \vopt_s\})}{\weight_t(\{\valg_t, \valg_s\})} \leq 2\sqrt{n}$.

    \medskip

    \runtitle{Claim 2: $\exclusiveWeight(\optDelayedSet(\vopt_t) \cup \optDelayedSet(\vopt_s)) \leq 2\sqrt{n}\cdot\weight(\charging(\partc_{t,s}))$.}
    Similar to the proof of Lemma~\ref{lem:typeBCycleCut}, we focus on the case where $\algCactus$ finishes the game in round $i_\ell$, which is earlier than \opt does. 
    Assume that \opt protects $\delta$ vertices between round $i_t$ (in which $\valg_t$ is protected) and round $i_\ell$.
    Since no vertex in $C$ is included in $\optDelayedSet(\vopt_t) \cup \optDelayedSet(\vopt_s)$, every vertex covers at most $\sqrt{\weight_t(\rootC)}$ vertices. 
    By construction, $\charging(\partc_{t, s})$ consists of all vertices in $\CoveredSet(\rootC_1) \cup \algset$. 
    Thus, $\charging(\partc_{t, s})$ covers at least $\amountBelow(\newT{C_t}{\valg_t}, i_\ell-i_t) + \weight_t(\valg_t)$ vertices.
    Meanwhile, during or after round $i_\ell$, $\optDelayedSet(\vopt_t)$ protects at most $\amountBelow(\newT{C}{\vopt_t}, i_\ell-i_t)$ vertices since the fire spreads for a distance of $i_\ell-i_t$. 
    Therefore, by Lemma~\ref{lem:CooldownIsFine}, $\frac{\exclusiveWeight(\optDelayedSet(\vopt_t) \cup \optDelayedSet(\vopt_s))}{\charging(\partc_{t, s})} \leq \frac{\delta\cdot \sqrt{\weight_t(\rootC_1)} + \amountBelow(\newT{C}{\vopt_t}, i_\ell-i_t)}{\delta + \amountBelow(\newT{C_t}{\valg_t, i_\ell-i_t}) + \weight_t(\valg_t)} \leq 2\sqrt{n}$.
\end{proof}

We will next focus our attention on the case where the cycle $C$, which contains the representative vertices, is not a cycle from the perspective of $\algCactus$.
That is, $\algCactus$ has, in an earlier round, protected or broken the cycle $C$.

\begin{lemma}\label{lem:cycleIsBroken}
    Given any Type-$b$ or Type-$c$ part $P$ where the representative vertex (or vertices) is in cycle $C$, if $\algCactus$ has broken or protected the cycle $C$ earlier than \opt does, $\exclusiveWeight(P) \leq 4 \sqrt{n} \cdot \weight(\charging(P))$.
\end{lemma}

\begin{proof}
    Let the (first) vertex in $C$ protected by $\opt$ have time label $t$. 
    We consider the graph $\newG_t$ that $\algCactus$ sees at time $t$, and the heaviest vertex $v_1$ in $\newG_t$.
    
    Since $P$ contains the covered set of its representative vertex (or vertices, in the case of Type-$c$ parts), whose representative lies in cycle $C$, no algorithm can save more vertices by protecting $P$ than by protecting $C$.
    Consequently, $\weight_t(C) \geq \weight_t(P)$.

    Let $s < t$ be the time when $\algCactus$ broke or protected the cycle $C$ by protecting $\valg_s$ or $\valg_s$ and $\valg_{s+1}$.
    In the latter case, part $P$ is in the covered set $\CoveredSet(\{\valg_s, \valg_{s+1} \})$ and therefore $\exclusiveWeight(P) = 0$.
    Hence we assume that $\algCactus$ chose to break the cycle at time $s$.
    The cycle $C$ is broken by $\algCactus$ into two arcs, $A_1$ and $A_2$ (at most one of them can be empty).
    We consider three possibilities of $\valg_t$: 1) $\valg_t$ is chosen as the heaviest vertex, 2) $\valg_t$ is one of the two vertices protecting $\rootC_1$, or 3) $\valg_t$ breaks a cycle. 

    If $\algCactus$ protects $v_1$ greedily, since the cycle-break at vertex $\valg_s$ has split $C$ into at most two arcs, $\weight_t(\valg_t) = \weight_t(v_1) \geq \max\{\weight_t(A_1), \weight_t(A_2)\} \geq \weight_t(C)/2 \geq \weight_t(P)/2 \geq \exclusiveWeight(P)/2$.
    
    Note that the second case only happens when $\weight_t(v_1) + \weight_t(v_2) < \weight_t(\rootC_1)$.
    The other vertex that protects $\rootC_1$ is either $\valg_{t-1}$ or $\valg_{t+1}$. 
    Assume without loss of generality that it is $\valg_{t+1}$.
    It follows that $\weight_t(\{ \valg_t, \valg_{t+1} \}) = \weight_t(\rootC_1) > \weight_t(v_1) + \weight_t(v_2) \geq \weight_t(\valg_s) + \weight_t(A_1) + \weight_t(A_2) = \weight_t(C) \geq \weight_t(P) \geq \exclusiveWeight(P)$.
    Where the second inequality follows from that the sum of the weights vertices on on the cycle $C$ in the original graph, that lie next to the root at time $t$ is at most the sum of the weights of the remaining parts of the arcs $A_1$ and $A_2$ and the weight of the vertex used to split the cycle $\valg_s$ (if this vertex is still unburned at time $t$).
    Recall that, as $\valg_t$ and $\valg_{t+1}$ are on the same cycle, and $\charging(P)$ contains all peers of $\valg_t$, it follows that $\valg_{t+1} \in \charging(P)$.
    Then $\weight(\charging(P)) \geq \weight_t(\{\valg_t, \valg_{t+1}\}) \geq \exclusiveWeight(P)$.

    In the third case, where $\algCactus$ calls \cactusBreakCycle, it follows from Lemma~\ref{lem:typeBCycleCut} and Lemma~\ref{lem:typeCCycleCut} that $\exclusiveWeight(P) \leq 4 \sqrt{n} \cdot \weight(\charging(P))$.
    
    Therefore, in either case, $\exclusiveWeight(P) \leq 4 \sqrt{n} \cdot \weight(\charging(P))$.
\end{proof}

\begin{lemma}[One cycle vertex with its delayed set]\label{lem:algCactusOneCVertex}
    If $\partb_{t}$ is a Type-$b$ part, then $\exclusiveWeight(\partb_{t}) \leq 4 \sqrt{n} \cdot \weight(\charging(\partb_{t}))$.
\end{lemma}

\begin{proof}
    Given a Type-$b$ part $\partb_t = \{\vopt_t\} \cup \optDelayedSet(\vopt_t)$, where $\optDelayedSet(\vopt_t)$ is the set of vertices in $\optset$ that are delayed by $\vopt_t$ (that is, $\optDelayedSet(\vopt_t) = \delayedSet(\vopt_t) \cup \optset$).
    Let $C$ be the cycle containing $\vopt_t$.
    Since $\vopt_t$ is the only vertex protected by \opt in $C$, $\weight_t(C) \geq \weight_t(\partb_t)$.    
    We consider the graph $\newG_t$ that $\algCactus$ sees at time $t$, where $v_1$ and $\rootC_1$ are the heaviest vertex and the heaviest cycle in $\newG_t$, respectively.
    By Lemma~\ref{lem:cycleIsBroken}, we focus on the case where the cycle $C$ is in $\newG_t$.
    
    By Observation~\ref{obs:algCactus}(\ref{obs:algCactusNoChain}) and the partition $\partition$, $\vopt_t$ is not delayed by $\opt$.
    Thus, $\vopt_t$ is available in $\newG_t$.
    Since every vertex in $\optDelayedSet$ is also in $\CoveredSet(C)$, 
    \begin{gather}\label{eq:PandMaxC}
    \exclusiveWeight(\partb_t) = \exclusiveWeight(\{\vopt_t\}\cup \optDelayedSet(\vopt_t)) \leq \weight_t(\{\vopt_t\}\cup \optDelayedSet(\vopt_t)) \leq \weight_t(C) \leq \weight_t(\rootC_1).
    \end{gather}
    There are three possibilities on how $\valg_t$ is protected by $\algCactus$: 1) $\valg_t$ is chosen as the heaviest vertex since $\valg_t = v_1$, 2) $\valg_t$ saves $\rootC_1$ together with $\valg_{t+1}$, or 3) $\valg_t$ is chosen by \cactusBreakCycle to break a cycle.

    The first case can be further split into three subcases: 
    \begin{itemize}
        \item If $\algCactus$ protects the heaviest vertex $v_1$ in $\newG_t$ because $\weight_t(v_1) \geq \sqrt{\weight_t(\rootC_1)}$, then $\weight(\charging(\partb_t)) \geq \weight_t(\charging(\partb_t)) \geq \weight_t(\valg_t) = \weight_t(v_1) \geq \sqrt{\weight_t(\rootC_1)}$.

        \item If $\algCactus$ protects the heaviest vertex $v_1$ because $\weight(\rootC_1) \leq \sqrt{n}$, then $\frac{\exclusiveWeight(\partb_t)}{\weight(\charging(\partb_t))} \leq \exclusiveWeight(\partb_t) \stackrel{\text{Eq. }(\ref{eq:PandMaxC})}\leq \weight_t({\rootC_1}) \leq \sqrt{n}$.

        \item $\algCactus$ can protect the heaviest vertex $v_1$ due to the non-zero cool-down counter that was set in an earlier time $s < t$ by $\valg_s$ breaking a cycle $C^\prime$.
        In this case, by construction, $\charging(\partb_t)$ contains $\valg_s$ and the vertices in $\CoveredSet_s(C^\prime) \cap \algset$.
        We make a case distinction based on whether root cycle $\rootC_1$ (in $\newG_t$) was already a root cycle in $\newG_s$.
        When $\rootC_1$ was not yet a root cycle in $\newG_s$, similar to the previous case, $\frac{\exclusiveWeight(\partb_t)}{\weight(\charging(\partb_t))} \leq \exclusiveWeight(\partb_t) \stackrel{\text{Eq. }(\ref{eq:PandMaxC})}\leq \weight_t({\rootC_1}) \leq \sqrt{n}$.
        
        The last inequality holds by contradiction:
        Suppose to the contrary that $\weight_t({\rootC_1}) > \sqrt{n}$.
        At time $s$ when $\algCactus$ broke cycle $C^\prime$, it could have also protected a vertex $v_k$ that is on the path from $r$ to $\rootC_1$.
        Hence $\weight(v_k) > \sqrt{n}$.
        Since $\algCactus$ chose to break cycle $C^\prime$, by line~\ref{Alg:cactus:line:greedy_with_one_firefighter} of $\algCactus$, $\sqrt{\weight(C^\prime)} > \weight(v_1) \geq \weight(v_k) > \sqrt{n}$.
        Thus $\weight(C^\prime) > n$, and we have reached a contradiction.

        Then, in the case where $\rootC_1$ was already a root cycle in $\newG_s$, we define the vertices $u_\ell$ as the left neighbor (on $\rootC_1$) of the root at time $s$, and $u_r$ as the right neighbor (on $\rootC_1$) of the root at time $t$.
        Then, by Lemma~\ref{lem:CactusBreakCycleIsGood}, $\frac{\exclusiveWeight(\partb_t)}{\weight_t(\charging(\partb_t))} \leq \frac{\weight_t(C)}{\weight(\{ \valg_s, \valg_t \})} \leq \frac{\weight_t(\rootC_1)}{\weight(\{ \valg_s, \valg_t \})} \leq \frac{\weight(\{ u_\ell, u_r \} )}{\weight(\{ \valg_s, \valg_t \})} \leq 2\sqrt{n}$.
    \end{itemize}

    Consider the case where $\valg_t$ protects the heaviest cycle $\rootC_1$ together with another vertex $\valg_{t+1}$.
    By construction, $\charging(\partb_t)$ also contains $\valg_{t+1}$. 
    Thus, $\weight(\charging(\partb_t)) \geq \weight(\rootC_1)$.
    Hence, by Inequality~(\ref{eq:PandMaxC}), $\exclusiveWeight(\partb_t) \leq \weight_t(\rootC_1) \leq \weight(\rootC_1) \leq \weight(\charging(\partb_t))$.

    For the third case, by Lemma~\ref{lem:typeBCycleCut}, $\exclusiveWeight(\partb_t) \leq 2\sqrt{n} \cdot  \weight(\charging(\partb_t))$.

    Together with Lemma~\ref{lem:cycleIsBroken}, the proof is completed.
\end{proof}

Similar to Lemma~\ref{lem:algCactusOneCVertex} and with a slight modification due to the different characterization of the parts, we prove the following bound for Type-$c$ parts:

\begin{lemma}[Two cycle vertices with their delayed set]\label{lem:algCactusTwoCVertices}
    If $\partc_{t, s}$ is a Type-$c$ part, then $\exclusiveWeight(\partc_{t, s}) \leq 4\sqrt{n}\cdot \weight(\charging(\partc_{t, s}))$.
\end{lemma}

\begin{proof}
    Given a Type-$c$ part $\partc_{t,s} = \{\vopt_t, \vopt_s\} \cup \optDelayedSet(\vopt_t) \cup \optDelayedSet(\vopt_s)$. 
    Without loss of generality, we assume that $t < s$.
    Let $C$ be the cycle containing $\vopt_t$ and $\vopt_s$.
    We consider the graph $\newG_t$ that $\algCactus$ sees at time $t$, where $v_1$ and $\rootC_1$ are the heaviest vertex and the heaviest cycle in $\newG_t$, respectively.
    By Lemma~\ref{lem:cycleIsBroken}, we focus on the case where the cycle $C$ is in $\newG_t$.

    Since every vertex in $\optDelayedSet(\vopt_t) \cup \optDelayedSet(\vopt_s)$ is also in the covered set of $C$, $\exclusiveWeight(\partc_{t,s}) = \exclusiveWeight(\{\vopt_t, \vopt_s\} \cup \optDelayedSet(\vopt_t) \cup \optDelayedSet(\vopt_s)) \leq \weight_t(C) \leq \weight_t(\rootC_1)$.

    There are three cases of $\valg_t$: 1) $\valg_t$ is selected as the heaviest vertex, 2) $\valg_t$ protects a cycle with another vertex $\valg_{t+1}$, or 3) $\valg_t$ breaks a cycle.
    The analysis of the first and the second cases is the same as the proof of Lemma~\ref{lem:algCactusOneCVertex}, where $\exclusiveWeight(\partc_{t,s}) \leq \sqrt{n} \cdot \weight(\charging(\partc_{t,s}))$.
    For the third case, where $\valg_t$ breaks a cycle $C_t$, it follows from Lemma~\ref{lem:typeCCycleCut} that $\exclusiveWeight(\partc_{t,s}) \leq 4\sqrt{n} \cdot \weight(\charging(\partc_{t,s}))$.
\end{proof}

\begin{proof}[Proof of Theorem~\ref{Thm:Cactus}]
    By Lemmas~\ref{lem:algCactusNonCycleVertices},~\ref{lem:algCactusOneCVertex}, and~\ref{lem:algCactusTwoCVertices}, for any part $P \in \partition$, $\exclusiveWeight(P) \leq 4\sqrt{n}\cdot \weight(\charging(P))$.
    Moreover, by Observation~\ref{obs:cactusBeta}, each vertex in $\algset$ is charged by at most $5$ times.
    Thus, by the non-redundancy of $\opt$, $\opt(\instance) \leq \alg(\instance) + (5 \cdot 4\sqrt{n})\alg(\instance) = (20\sqrt{n}+1)\alg(\instance)$.
    
    The optimality follows from Theorem~\ref{thm:TadpoleLB}.
\end{proof}

\section{Cactus graphs with even firefighters each round}
\label{Sec:even}
In this section, we extend our framework to a special case of online firefighting on cactus graphs, where, in each round, either there is no firefighter available or an even number of firefighters are available.
We show that a greedy algorithm modified from our algorithm for cactus graphs is exactly $3$-competitive.

\subsection{The greedy algorithm $\algEven$}
We follow the graph reduction framework.
In each round, the algorithm $\algEven$ places the firefighters in a way that maximizes the number of vertices immediately saved during this round.
More specifically, when two or more firefighters are available, $\algEven$ protects the heaviest cycle $\rootC_1$ if it is heavier than the total weight of the two heaviest vertices $v_1$ and $v_2$ in the graph.
Otherwise, the algorithm protects the heaviest vertex $v_1$.
When there is only one firefighter available, the algorithm protects $v_1$.
(See Algorithm~\ref{Alg:even} for details.)
Note that although each round an even number of firefighters (or none) become available, due to the greedy nature of our algorithm, which might choose to protect $v_1$ first, in the middle of a round the algorithm might be left with an odd number of firefighters.

\begin{algorithm}[t]
    \caption{Procedure of $\algEven$ for cactus graph with even firefighters in round $i$} 
    \label{Alg:even}
    \textbf{Input:} Graph $\newG$ with fire source $\newR$ and $f_i$ available firefighters\\
    \While{$f_i > 0$ and there are available vertices in $\newG$}{
        \eIf{$\newR$ is not a cycle vertex in $\newG$}{
            Protect the vertex $v$ in $N(\newR)$ with the highest weight.\\
            Set $f_i\gets f_i -1$, $\newG \gets \newG \setminus \CoveredSet(v)$.
        }{
            \tcp{There are root cycles $\rootC_1, \rootC_2, \cdots$ containing $\newR$ and $\weight(\rootC_1) \geq \weight(\rootC_2) \geq \cdots$}
            Consider unprotected vertices $v_1, v_2, \cdots$ in $N(\newR)\cup (\bigcup_h\rootC_h)$ with $\weight(v_1) \geq \weight(v_2) \geq \cdots$

            \eIf{$f_i \geq 2$}{
                \eIf{$\weight(v_1) + \weight(v_2) \geq \weight(\rootC_1)$}{   
                    Protect $v_1$.\\
                    Set $f_i \gets f_i - 1$, $\newG \gets \newG \setminus \CoveredSet(v_1)$.
                }{
                    Protect both vertices $v_\ell, v_r$ in $N(\newR) \cap \rootC_1$.\\ Set $f_i \gets f_i - 2$, $\newG \gets \newG \setminus \CoveredSet(\{ v_\ell, v_r \})$.
                }
            }{
                Protect $v_1$.\\
                Set $f_i \gets f_i - 1$, $\newG \gets \newG \setminus \CoveredSet(v_1)$.
            }
        }
    }
\end{algorithm}

\subsection{Competitive analysis}
Given the sets of vertices $\algset$ and $\optset$ protected by $\algEven$ and by a non-redundant optimal solution, respectively, we follow the analysis framework introduced in Section~\ref{Sec:Framework}.
Note that under algorithm $\algEven$, a cycle can be protected by vertices $\valg_t$ and $\valg_{t+1}$ in $\algset$, with $t$ being either odd or even.

\runtitle{Partition $\partition$ of $\optset$.}
We partition the vertices in $\optset$ into three types of parts, which is an adaption from the partition for analyzing $\algCactus$:
\begin{enumerate}[Type a:]
    \item \runtitle{One non-cycle vertex that is not delayed.}
    For any non-cycle vertex $\vopt_t \in \optset$ that is not delayed by $\optset$, there is a part $\parta_t = \{\vopt_t\}$, which is a singleton, and $\vopt_t$ is the representative vertex in $\parta_t$.

    \item \runtitle{One cycle vertex that is not delayed.}
    For each cycle vertex $\vopt_t \in \optset$ that is the only vertex in the cycle protected by $\optset$, there is a Type-$b$ part $\partb_t = \{\vopt_t\} \cup \optDelayedSet(\vopt_t)$.
    We call $\vopt_t$ the \emph{representative vertex} of $\partb_t$.
    
    \item \runtitle{Two cycle vertices with their delayed set.}
    For each cycle $C$ that contains two vertices $\vopt_t$ and $\vopt_{s}$, and at least one of $\vopt_t$ and $\vopt_{s}$ is not delayed by $\optset$, there is a part $\partc_{t,s} = \{\vopt_t, \vopt_{s}\} \cup \optDelayedSet(\vopt_t) \cup \optDelayedSet(\vopt_{s})$.
    Recall that the index of $\vopt_{\min \{ t, s \} }$ is the time label in our framework.
    We call $\vopt_{\min\{t,s\}}$ the \emph{representative vertex} of $\partc_{t,s}$.
\end{enumerate}

Note that, unlike the partition used in the analysis of the algorithm for cactus graphs, there may be multiple Type-$a$ parts.

\medskip

\runtitle{Charging function.}
Given any part $P$, our charging function maps every representative vertex $\vopt_t$ to $\valg_x$ and $\valg_{x+1}$, where $x$ is the \emph{anchor} of $P$, defined as the largest odd number which is at most $t$ (where $t$ is the first vertex in $P$ protected by $\opt$).
That is, for a Type-$a$ part $\parta$ or a Type-$b$ part $\partb$, $\charging(P^{(\cdot)}_{t}) = \{\valg_{x}, \valg_{x+1} \mid x \text{ is the anchor of }P^{(\cdot)}_{t}\}$.
For a Type-$c$ part $\partc_{t,s}$, $\charging(\partc_{t,s}) = \{\valg_{x}, \valg_{x+1} \mid x \text{ is the largest odd number and } x \leq \min\{t, s\}\}$.

\paragraph*{Proof of Theorem~\ref{Thm:Even}}
\begin{proof}
    Given a part $P$ with anchor time $x$, we consider the reduced graph $\newG_x$ with fire source $\newR$ from $\algEven$'s perspective at time $x$. 
    We denote $v_1$ the heaviest vertex in $\newG_x$, $v_2$ the second heaviest vertex in $\newG_x$, and $\rootC_1$ the heaviest root cycle in $\newG_x$.
    If there is no root cycle in $\newG_x$, we let $\rootC_1 = \emptyset$ and $\weight(\rootC_1) = 0$.
    Note that since each round $i$ has even firefighters and $x$ is odd, at time $x$ there are at least two firefighters available to $\algEven$. 

    \medskip

    \runtitle{Given a part $P$ with anchor time $x$, $\exclusiveWeight(P) \leq \weight_x(v_1)$ or $\exclusiveWeight(P)\leq \weight_x(\rootC_1)$:}
    \begin{itemize}
        \item \runtitle{Type-$a$ part $\parta_t$.}
        Given the anchor time $x$ of $\parta_t = \{\vopt_t\}$, under the algorithm, $\exclusiveWeight(\parta_t) = \weight_x(\vopt_t) \leq \weight_x(v_1)$.

        \item \runtitle{Type-$b$ part $\partb_t$.}
        Given the anchor time $x$ of $\partb_t = \{\vopt_t\} \cup \optDelayedSet(\vopt_t)$, let $\rootC$ be the root cycle containing $\vopt_t$. 
        Every vertex in $\optDelayedSet(\vopt_t)$ lies in $\CoveredSet(\rootC)$.
        Thus, $\exclusiveWeight(\partb_t) \leq \exclusiveWeight(\rootC) \leq \weight_x(\rootC) \leq \weight_x(\rootC_1)$.

        \item \runtitle{Type-$c$ part $\partc_{t,s}$.}
        Given the anchor time $x$ of $\partc_{t, s} = \{\vopt_t, \vopt_s\} \cup \optDelayedSet(\vopt_t) \cup \optDelayedSet(\vopt_s)$, let $\rootC$ be the root cycle containing $\vopt_t$ and $\vopt_s$.
        Since every vertex in $\optDelayedSet(\vopt_t)$ or $\optDelayedSet(\vopt_s)$ lies in $\CoveredSet(\rootC)$.
        Thus, $\exclusiveWeight(\partc_{t, s}) \leq \exclusiveWeight(\rootC) \leq \weight_x(\rootC) \leq \weight_x(\rootC_1)$.
    \end{itemize}
    
    \medskip

    For any $\valg_x \in \algset$, we define a variant of weight function: $\newWeight(\valg_x) = \frac{1}{2}\weight_{\min\{x,y\}}(\rootC_1)$ if $\valg_x$ is chosen to protect cycle $\rootC_1$ together with another vertex $\valg_y$. Otherwise, $\newWeight(\valg_x) = \weight_x(\valg_x)$.
    By construction, $\sum_{x} \newWeight(\valg_x) = \weight(\algset)$.
    We further define that given $S \subseteq \algset$, $\newWeight(S) := \sum_{\valg \in S} \newWeight(\valg)$. 
    Thus, $\newWeight(S) \geq \newWeight(\valg)$ for any $\valg \in S$.
    Note that it is different from the definition of $\weight(S)$, which is the size of $\CoveredSet(S)$.

    \medskip
    
    \runtitle{Given any part $P$ with anchor time $x$, $\weight_x(\rootC_1) \leq \newWeight(\charging(P))$ and $\weight_x(v_1) \leq \newWeight(\charging(P))$:}
    Consider the vertices $\valg_x$ and $\valg_{x+1}$ protected by $\algEven$ at any two consecutive times $x$ and $x+1$. 
    Recall that under the algorithm, any vertex can be chosen to protect the heaviest vertex $v_1$ or to protect ``half'' of the heaviest cycle $\rootC_1$. 
    Therefore, there are five cases:
    \begin{enumerate}[\text{Case} (i)]
        \item $\valg_x = v_1$ and $\valg_{x+1} = v_2$: This case happens when $\weight_x(v_1) + \weight_x(v_2) \geq \weight_x(\rootC_1)$.
        Thus, $\weight_x(v_1) = \newWeight(\valg_x) \leq \newWeight(\{\valg_x, \valg_{x+1}\})$, and $\weight_x(\rootC_1) \leq \weight_x(v_1) + \weight_x(v_2) = \newWeight(\valg_x) + \newWeight(\valg_{x+1}) = \newWeight(\{\valg_x, \valg_{x+1}\})$.

        \item $\valg_x = v_1$ and $\valg_{x+1}$ is protecting $\rootC_1$ (with $\valg_{x+2})$: It is because $\weight_x(v_1) + \weight_x(v_2) \geq \weight_x(\rootC_1)$.
        Thus, $\weight_x(v_1) = \newWeight(\valg_x) \leq \newWeight(\{\valg_x, \valg_{x+1}\})$.
        Moreover, since $\weight_x(v_1) + \weight_x(v_2) \geq \weight_x(\rootC_1)$, $\newWeight(\valg_x)=\weight_x(v_1) \geq \frac{1}{2}\weight_x(\rootC_1)$.
        Thus, $\weight_x(\rootC_1) \leq \newWeight(\valg_x)+\newWeight(\valg_{x+1}) = \newWeight(\{\valg_x, \valg_{x+1}\})$.

        \item $\valg_x$ is protecting cycle $\rootC_1$ with $\valg_{x+1}$: It is because $\weight_x(\rootC_1) > \weight_x(v_1) + \weight_x(v_2)$.
        In this case, $\weight_x(v_1) < \weight_x(\rootC_1) = \newWeight(\valg_x) + \newWeight(\valg_{x+1}) = \newWeight(\{\valg_x, \valg_{x+1}\})$.

        \item $\valg_x$ is protecting cycle $\rootC_0$ with $\valg_{x-1}$, and $\valg_{x+1} = v_1$: 
        In this case, $\weight_x(v_1) = \newWeight(\valg_{x+1}) \leq \newWeight(\{\valg_x, \valg_{x+1}\})$.
        Moreover, this case happens because $\weight_{x-1}(\rootC_0) \geq \weight_{x-1}(\rootC_1) = \weight_x(\rootC_1)$ and $\weight_x(v_1) + \weight_x(v_2) \geq \weight_x(\rootC_1)$. 
        Thus, $\weight_x(v_1) \geq \frac{1}{2}\weight_x(\rootC_1)$ since $\weight_x(v_1) \geq \weight_x(v_2)$.
        Therefore, $\newWeight(\{\valg_x, \valg_{x+1}\}) = \newWeight(\valg_x) + \newWeight(\valg_{x+1}) = \frac{1}{2}\weight_{x-1}(\rootC_0) + \weight_x(v_1) \geq \frac{1}{2}\weight_{x}(\rootC_1) + \frac{1}{2}\weight_x(\rootC_1) = \weight_x(\rootC_1)$.

        \item $\valg_x$ is protecting cycle $\rootC_0$ with $\valg_{x-1}$, and $\valg_{x+1}$ is protecting cycle $\rootC_1$ with $\valg_{x+2}$: 
        It is because $\weight_{x-1}(\rootC_0) \geq \weight_{x-1}(\rootC_1) = \weight_x(\rootC_1)$ and $\weight_x(\rootC_1) > \weight_x(v_1) + \weight_x(v_2)$. 
        In this case, $\newWeight(\{\valg_x, \valg_{x+1}\}) = \newWeight(\valg_x) + \newWeight(\valg_{x+1}) = \frac{1}{2}\weight_{x-1}(\rootC_0) + \frac{1}{2}\weight_x(\rootC_1) \geq \frac{1}{2}\weight_{x}(\rootC_1) + \frac{1}{2}\weight_{x}(\rootC_1) = \weight_x(\rootC_1) > \weight_x(v_1)$.
    \end{enumerate}
    Therefore, we can conclude that for any part $P$ with anchor time $x$, $\exclusiveWeight(P) \leq \newWeight(\charging(P))$ since $\charging(P) = \{\valg_x, \valg_{x+1}\}$. 
    
    Moreover, since the charging function $\charging$ charges each part $P$ to $\valg_x$ and $\valg_{x+1}$ where $x$ is the anchor time of $P$, each $\valg_t$ can be charged at most twice.
    Hence, by our framework, $\frac{\opt(\instance)}{\algEven(\instance)} = \frac{\weight(\optset)}{\weight(\algset)} = \frac{\exclusiveWeight(\optset)}{\weight(\algset)} + 1 = \frac{\sum_{P} \exclusiveWeight(P)}{\weight(\algset)}+1 \leq \frac{\sum_{P} \weight(\charging(P))}{\weight(\algset)} + 1 \leq \frac{2\weight(\algset)}{\weight(\algset)} + 1 = 3$.
\end{proof}

\begin{figure}[bt]
    \centering
    \includegraphics[width=0.5\linewidth]{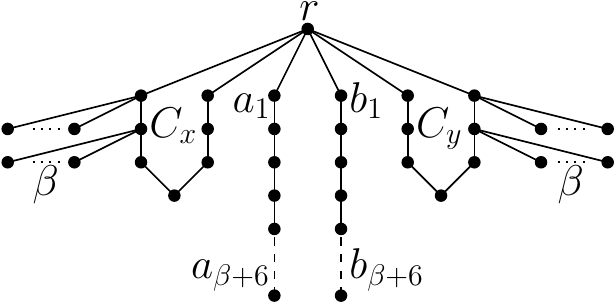}
    \caption{The construction of the graph where $\algEven$ is exactly $3$-competitive.}
    \label{fig:thm27}
\end{figure}

\begin{theorem}
    The algorithm $\algEven$ is exactly $3$-competitive.
\end{theorem}

\begin{proof}
    We consider the input graph $G$ as follows:
    There is two cycles with $8$ vertices, $C_x = (r, x_1, x_2, \cdots, x_7)$ and $C_y = (r, y_1, y_2, \cdots, y_7)$, where $r$ is the fire source.
    In the first cycle, there are $\beta$ degree-$1$ vertices adjacent to $x_1$, and another $\beta$ degree-$1$ vertices adjacent to $x_2$.
    Similarly, in the second cycle, there are $\beta$ degree-$1$ vertices adjacent to $y_1$, and another $\beta$ degree-$1$ vertices adjacent to $y_2$.
    There are two other paths $r, a_1, a_2, \cdots, a_{\beta+6}$ and $(r, b_1, b_2, \cdots, b_{\beta+6})$.
    The firefighter sequence is $(2, 0, 0, 0, 4, 0, \cdots)$.
    That is, $f_1 = 2$, $f_5 = 4$, and $f_i = 0$ for all $i \notin \{1, 5\}$.

    In the original graph, the heaviest vertices are $a_1$ and $b_1$, and the heaviest cycle $C_x$ or $C_y$ has weights of $7+2\beta$, which is not heavier than $\weight(a_1) + \weight(b_1) = 2\beta + 12$.
    Thus, in round $1$, $\algEven$ protects $a_1$ in the first iteration.
    In the second iteration, the current heaviest vertex is $b_1$, and $\algEven$ protects $b_1$.
    Subsequently, at the beginning of round $5$, both cycles have been burnt entirely. 
    Hence, the $\algEven$ saves $2\beta+12$ vertices.

    On the other hand, an algorithm can protect the vertices $a_1$ and $b_1$ in the first round, and protect $a_5$, $b_5$, $x_3$, and $y_3$ in round $5$.
    Thus, $\opt$ saves at least $(\beta + 1) + (\beta + 1) + (\beta + 2) + (\beta + 2) + (\beta + 2) + (\beta + 2) = 6 \beta + 10$.
    By selecting an arbitrarily large $\beta$, the ratio tends to $3$.
\end{proof}

\newpage

\bibliography{ref}
\end{document}